\documentclass[onecolumn,draft]{IEEEtran}
\ifCLASSINFOpdf
\else
\fi
\usepackage{amssymb}
\usepackage{amsmath}
\usepackage{cite}
\usepackage{graphicx}
\usepackage{subfigure}
\usepackage{epstopdf}
\usepackage{algorithm}
\usepackage{algorithmic}

\usepackage{tikz}
\usetikzlibrary{arrows,automata}

\usepackage{lscape}
\usepackage{graphicx}
\usepackage{fancyhdr}
\usepackage{tikz}
\usepackage{epstopdf}
\usetikzlibrary{automata}
\usepackage{array}
\usepackage{multirow}
\usepackage{amsmath,amsthm}
\usepackage{footnote}
\usepackage{supertabular}
\usepackage{mathtools}
\usepackage{amssymb}
\usepackage{longtable}
\usepackage{nomencl}
\usepackage[nonumberlist,acronym]{glossaries}
\usepackage{algorithmic}
\usepackage{color}
\usepackage{algorithm}
\usepackage{enumerate}
\usepackage{datatool}
\usepackage{amsmath}
\usepackage{nomencl}

\DeclareMathOperator*{\mslim}{l{\cdot}i{\cdot}m}

\begin{document}
%
\title{Noise Models in the Nonlinear Spectral Domain for Optical Fibre Communications}
%
%
%

\author{Qun Zhang 
        and Terence H. Chan,~\IEEEmembership{Member,~IEEE}
\thanks{
This paper was presented in part at 2015 IEEE International Symposium on Information
Theory (ISIT 2015).}
}

\newcommand\ter[1]{{\color{red} #1 }}

\maketitle

\begin{abstract}
Existing works on building a soliton transmission system only encode information using the imaginary part  of the eigenvalue, which fails to make full use of the signal degree-of-freedoms. 
Motivated by this observation, we make the first step of encoding information using (discrete) spectral amplitudes by proposing analytical noise models for the spectral amplitudes of $N$-solitons ($N\geq 1$). To our best knowledge, this is the first work in building an analytical noise model for spectral amplitudes, which leads to many interesting information theoretic questions, such as channel capacity analysis, and has a potential of increasing the transmission rate. The noise statistics of the spectral amplitude of a soliton are also obtained without the Gaussian approximation. 
\end{abstract}


\newtheorem{df}{Definition}
\newtheorem{thm}{Theorem}
\newtheorem{prop}{Proposition}
\newtheorem{lemma}{Lemma}
\newtheorem{example}{Example}
\newtheorem{cor}{Corollary}
\newtheorem{rem}{Remark}
\newtheorem{conjecture}{Conjecture}

\newcommand{\seq}[3]{{#1_{#2}, \ldots, #1_{#3}}}
\newcommand{\nge}[1]{\stackrel{#1}{\ge}}
\newcommand{\nequal}[1]{\stackrel{#1}{=}}
\def\A{\mathcal A}
\def\B{\mathcal B}
\def\C{\mathsf C}
\def\M{\mathcal M}
\def\N{\mathcal N}
\def\V{\mathcal V}
\def\X{\mathcal X}
\def\Y{\mathcal Y}
\def\Z{\mathcal Z}
\def\reals{\mathbb R}
\def\p{\prime}

\def\Mi{{B}}
\def\hG{{\widehat G}}
\def\ft{{\Xi}}
\def\ba{{\bf a}}
\def\bA{{\bf A}}
\def\bb{{\bf b}}
\def\bz{{\bf z}}
\def\id{{\bf e}}
\def\br{{\bf r}}
\def\bs{{\bf s}}
\def\bu{{\bf u}}
\def\B{{\mathcal B}}

\def\ed{\marginpar{\footnotesize$ \:$\vspace{-2cm}}}
\def\real{{\mathbb R}}

\def\defined{\: \triangleq \:}

\def\C{{\cal C}}
\def\M{{\cal M}}
\def\X{{\bf X}}
\def\Y{{\bf Y}}
\def\d{{\mathrm{d}}}

\def\im{{\mathrm{Im}}}
\def\re{{\mathrm{Re}}}

\def\sch{{Schr\"{o}dinger\:}}
\def\length{\mathfrak L}
\newcommand{\tc}[1]{  {\color{blue} {\bf  } #1}}

\def\P{{\mathcal P}}

\def\sch{{Schrodinger\:}}
\def\length{\mathfrak L}
\def\re{\mathrm {Re}}
\def\im{\mathrm {Im}}
\def\nri{\Upsilon^{(R)}_{i}}
\def\nr{\Upsilon^{(R)}}
\def\nii{\Upsilon^{(I)}_{i}}
\def\ni{\Upsilon^{(I)}}
\def\nid{\Upsilon_{i}}
\def\n{\Upsilon}
\def\nrg{\Upsilon^{(R)}_{G}}
\def\nig{\Upsilon^{(I)}_{G}}
\def\ng{\Upsilon_{G}}
\def\ito{It\^{o}}

%
\IEEEpeerreviewmaketitle

\section{Introduction}
%
%
%
%
\IEEEPARstart{O}{ptical} fibres are promising media for high speed data transmission because of the ultra high bandwidth and low loss transmission they offer. Compared to traditional linear frequency channels, signal propagation therein is nonlinear, and is described by the It\^{o} stochastic nonlinear Schr\"{o}dinger equation (SNLSE) \cite{b4,DBLP:journals/corr/abs-1202-3653} 
\begin{multline}
\frac{\partial A(s,l)}{\partial l}-\frac{j\beta_{2}}{2}\frac{\partial^{2}A(s,l)}{\partial s^{2}}+\frac{\alpha}{2}A(s,l)
=-j\gamma |A(s,l)|^{2}A(s,l)+j\kappa N(s,l),\qquad 0\leq l\leq \mathfrak{L}\ \mathrm{km},\label{ch3.1.3}
\end{multline}
where $j=\sqrt{-1}$, 
and $\mathfrak{L}>0$ denotes the length of the optical fiber. The second and third terms on the left hand side of \eqref{ch3.1.3} represent respectively the group velocity dispersion (GVD) and fibre attenuation effects, and the latter is often eliminated under the assumption that the fibre loss is perfectly compensated by the ideal distributed Raman amplification (DRA).   
The first term on the right hand side of \eqref{ch3.1.3} means fibre nonlinearity. The optical noise field $j\kappa N(\tau,l)$ was shown \cite{2010JLT/EKWFG} to be dominated by the amplified spontaneous emission (ASE) produced by the DRA when the fibre loss was compensated, and could be represented as a zero mean circularly symmetric complex white Gaussian noise \cite{1963GLW,1963GWL,1994M,b29,b41}, 
where 
$\kappa^{2}=\alpha h\nu_{s}K_{T}$. The parameters are summarised in Table \ref{ch1.table1.1} obtained from \cite{DBLP:journals/corr/abs-1302-2875,2010JLT/EKWFG}.
\begin{table*}[!t]
\renewcommand{\arraystretch}{1.3}
\caption{Fibre Parameters}
\label{ch1.table1.1}
\centering
\begin{tabular}{|c|c|c|}
\hline
Symbols & Values & Explanations\\
\hline
$\alpha$ & $0.046$ km$^{-1}$ & Fibre loss (0.2 dB/km) \\
$h$ & $6.626\times 10^{-34}$ J$\cdot$s & Planck's constant \\
$\nu_{s}$ & $193.55$ THz & Centre frequency \\
$K_{T}$ & $1.13$ & Photon occupancy factor 
\\
$\gamma$ & $1.27$ W$^{-1}\cdot$km$^{-1}$ & Nonlinear parameter\\
$\beta_{2}$ & $-2\times 10^{-23}$ s$^{2}\cdot$km$^{-1}$ & The GVD coefficient for silica fibres when the input wavelength is near $1.5$ $\mathrm{\mu m}$\\
\hline
\end{tabular} 
\end{table*}
After applying the following variable transformations
\begin{equation}\label{ch2.x1.1}
{q=\frac{A}{\sqrt{P_{n}}},\qquad t=\frac{s}{T_{n}},\qquad z=\frac{l}{L_{n}},}
\end{equation}
where
\begin{equation}\label{ch2.x1.2}
P_{n}=\frac{2}{\gamma L_{n}},\qquad T_{n}=\sqrt{\frac{|\beta_{2}|L_{n}}{2}},
\end{equation}
and $L_{n}$ is chosen to be $L_{n}=1$ $\mathrm{km}$, we obtain the normalised SNLSE
\begin{equation}\label{ch3.1.1}
{j\frac{\partial q(t,z)}{\partial z}=\frac{\partial^{2} q(t,z)}{\partial t^{2}}+2|q(t,z)|^{2}q(t,z)+j\epsilon G(t,z).}
\end{equation}
The parameters $t$ and $z$ are respectively the normalised temporal and spatial variables. The (normalised) optical noise field $G(t,z)$ is a zero mean circularly symmetric complex white Gaussian noise process
\begin{equation}\label{ch3.1.7}
{\mathrm{E}\left[G(t,z)G^{*}(t',z')\right]=\delta(t-t')\delta(z-z')}
\end{equation}
with noise power spectral density $\epsilon^{2}=\frac{\gamma}{\sqrt{2|\beta_{2}|}}\kappa^{2}$, where we use ``$*$'' to denote the complex conjugate, and $\delta(x)$ is the Dirac delta function. 
The normalised SNLSE \eqref{ch3.1.1} defines a noisy nonlinear dispersive waveform channel, in which $q(t,0)$ is the channel input, and $q(t,\length)$ is the channel output.

The GVD and Kerr nonlinearity have made the channel rather analytically difficult to tract compared to a radio frequency channel. Because of this, only simple modulation and coding schemes, such as on-off keying  (OOK) \cite{1998CTCCL} and differential phase shift keying \cite{2012JMMLC}, were employed in current optical fibre communication systems. Multiplexing techniques were applied to allow multiuser communications, such as wavelength division multiplexing (WDM), however, the performance is limited by the inter-channel interference. In many cases, nonlinear dispersive effects were often ignored for modelling simplicity. As such, a fibre channel would either be modelled as linear or dispersion free \cite{2011YK}, and  communication techniques designed for linear systems such as digital backpropagation and dispersion management \cite{2008JLT/IK} would be adopted. These techniques run well in a low signal power regime, as shown in \cite{2010JLT/EKWFG} that the capacity increases by increasing the signal power at first, however, then reaches a peak and vanishes as the power of signal further increases.  

Rather than regarding the dispersion and nonlinear effects as detrimental, they could be fully characterised by applying the nonlinear Fourier transform (NFT) \cite{1972JETP/ZS,b3} (also called direct scattering transform abbreviated as DST), a powerful tool to solve integrable nonlinear Schr\"{o}dinger equations (NLSE). By using the NFT, it was pointed out \cite{1993HN} that the eigenvalues of a signal could be used for data transmission in an optical fibre for their desirable property of remaining constants in the noise free case although its time domain representation is distorted. Single-user channels were then proposed, in which the information was modulated by those invariant quantities, and the decoding was achieved by applying the NFT. Recently, Yousefi and Kschischang published a series of papers \cite{DBLP:journals/corr/abs-1202-3653,DBLP:journals/corr/abs-1204-0830,DBLP:journals/corr/abs-1302-2875} aiming to diagonalise the time domain dispersive nonlinear channel into multiple linear scalar multiplicative channels in the spectral domain with the help of the NFT. In \cite{DBLP:journals/corr/abs-1202-3653}, the mathematical details of the NFT were introduced. In \cite{DBLP:journals/corr/abs-1204-0830}, several numerical algorithms implementing the NFT were proposed. In \cite{DBLP:journals/corr/abs-1302-2875}, the spectral domain modulation was considered, and the spectral efficiency for $N$-soliton communication systems was studied, where $N$ is a positive integer.


In those papers, a novel multiplexing technique, namely nonlinear frequency division multiplexing (NFDM), was proposed, which eliminates the inter-channel interference from multiple users for the multiuser communication scenario when noise is absent, and fully characterises the dispersion effect and fibre nonlinearity. The idea is similar to the orthogonal frequency division multiplexing (OFDM), in which the interference is avoided in the frequency domain by properly allocating the bandwidth to users with the help of the Fourier transform (FT) \cite{b17}. In the NFDM technique, the FT is replaced by the NFT with the transformed domain becoming (nonlinear) spectral domain accordingly, but the idea of eliminating the inter-channel interference is essentially similar. With the help of the NFT, the signal input-and-output relationship in the spectral domain becomes linear scalar multiplications, which is significantly simpler than that in the time domain. Recently, the NFDM was demonstrated experimentally the modulation and error-free detection of the eigenvalues of some $N$-solitons ($N=1,2,3,4$) \cite{2015DHGZYLWKL}. Note that in the NFDM scheme, the signal modulation and demodulation relies respectively on the inverse nonlinear Fourier transform (INFT) and NFT that are computationally expensive. Fast algorithms implementing them were therefore investigated by Wahls and Poor \cite{201305WP,2014arXiv1402.1605SV,201407SV,201506WP}.

%

Fibre optics communication systems using $N$-solitons ($N\geq 1$) are of great interest over the past few decades. When the inputs are (fundamental) solitons ($N=1$), the effect of noise in soliton parameters was studied. In particular, the statistics of an eigenvalue was reported in \cite{b30,b31}. The arrival time jitter characterising the fluctuations of the arrival time of a soliton, namely Gordon-Haus effect, was studied in the celebrated paper \cite{1986GH}. The research about soliton transmission control, aiming to tackle the issue of timing jitter, could be found in \cite{1991MMHL,1992MGE,1991NYKS,1995M}. The Gordon-Mollenauer effect, referring to the soliton phase jitter, was investigated in \cite{1990GM}, and the work about its statistics in soliton-dispersion phase shift keying systems were studied in \cite{1999HPGR,2000HPGR,2001HPGR,2002MX}. Under the NFT transmission framework, the capacity results of soliton communication systems were studied \cite{DBLP:journals/corr/abs-1302-2875,MeronTAU2012,2012arXiv1207.0297M}. These works mainly focused on modulating the eigenvalues of the signals only, and assumed a conditional Gaussian distributed noise as an approximation. Falkovich \emph{et.al.} \cite{2001FKLT} studied the statistics of the soliton parameters without the Gaussian approximation, and Derevyanko \emph{et.al.} \cite{2003DTY,2005DTY} applied the Fokker-Planck equation approach to obtain the marginal conditional distribution of the frequency and the amplitude of a soliton, or equivalently, the real and the imaginary parts of its eigenvalue, respectively. A lower bound of the channel capacity was derived for soliton communication systems when only soliton amplitude (the imaginary part of the eigenvalue) was modulated using the exponential distribution as the input distribution with its non-Gaussian conditional probability density function of the output given an input \cite{2015ITW2/SPDABT}. For general $N$-soliton inputs ($N\geq 2$), communication using the NFT was studied recently in \cite{2016JLT/HYK}.

We notice that existing $N$-soliton communication systems  are quite restrictive \cite{DBLP:journals/corr/abs-1302-2875} in the sense that information is only conveyed via the eigenvalues\footnote{To be precise, only the imaginary part of the eigenvalues are modulated.} (but not the spectral amplitudes) of the input signals. Hence, the capacity \cite{b14} of the fibre channel is not fully utilized although the alphabet of the imaginary part of the eigenvalue is generalized to a continuous interval compared to the traditional on-off keying transmission scheme, in which we either transmit a bit ($1$) or nothing ($0$).
 We find out that encoding information using the spectral amplitude is important, which increases the signal degree-of-freedom, and further increases the data rate. 
In this paper, we are interested in how to convey information by also modulating the spectral amplitudes of the input signals. To achieve this goal, a channel model characterising the statistical relations between the channel input and output of the spectral amplitudes is needed. One of the main contributions of this paper is to derive such an analytical channel model for spectral amplitudes. 

We also notice that to model the perturbation of the eigenvalue of a soilton as Gaussian distributed \cite{DBLP:journals/corr/abs-1302-2875,MeronTAU2012} is not precise. Specifically, modelling the noise in an eigenvalue as Gaussian is an approximation, as also pointed out in \cite{DBLP:journals/corr/abs-1302-2875}. 
In this paper, we do not make such approximation, and derive a non-Gaussian analytical noise model of the spectral amplitudes of solitons ($N$-solitons). The non-Gaussian noise statistics are also obtained.

The rest of this paper is organised as follows. In Section \ref{Contributions}, the contributions of this paper are summarised. In Section \ref{Preli}, some preliminaries of this paper is introduced, such as the basis of NFT and the idea of the NFDM. Sections \ref{Model} and \ref{Special} present the main results of this paper. We propose channel models of the (discrete) spectral amplitude (both its magnitude and phase) of $N$-solitons in Section \ref{Model}. We then study our channel model in Section \ref{Special} for a special case, i.e. soliton inputs, where more tools, such as perturbation theory, are available. Besides the channel model, the non-Gaussian noise statistics are also derived analytically in this case. In addition, we discuss our modelling methodology in Section \ref{Discussions} with the help of the perturbation theory, and show that the
noise captured in our model is significant, or even dominant, in some scenarios such as long distance transmission at a high input power. Section \ref{Conclusions} concludes this paper, and points out future works. Some of the proofs, either too long to be included in the main body or just about mathematical tools, are put in the appendices.

\section{Contributions}\label{Contributions}

In this paper,  we propose an analytical noise model for spectral amplitudes of $N$-solitons \cite{Zhan1406:Spatial}, where distributed white Gaussian noise (due to the DRA) is assumed in the nonlinear Schr\"{o}dinger channel \eqref{ch3.1.1}. To our best knowledge, this is the first work of building a mathematical noise model of spectral amplitudes of $N$-solitons in the presence of GVD and fibre nonlinearity.  

Specifically, we present an important and useful modelling methodology for the distributed white Gaussian noise process in \eqref{ch3.1.1}  in Section \ref{Model}. The method is called  \emph{split and concatenate}, motivated by the idea of It\^o integration. To be more precise, an optical fibre channel is split into many ultra short segments such that the distributed white Gaussian noise could be treated as lumped in each segment, and the transmission thereafter within this segment is noiseless. With all these segments concatenated, the model is obtained by taking the mean-square limit, which we mimic the characterisation of the distributed noise on the spectral amplitude.

Our work is a first step towards the ultimate goal of modulating spectral amplitudes in $N$-soliton communication systems in order to have a better use of the signal degree-of-freedoms (and hence a higher data transmission rate). 
The spectral domain representations of an $N$-soliton is composed by its eigenvalues and spectral amplitudes. While previous works mainly focus on encoding bits using eigenvalues but not spectral amplitudes \cite{DBLP:journals/corr/abs-1302-2875,2012arXiv1207.0297M}, our work provides a starting points of increasing the transmission rate. 
We also note that our model is an analytical one, which has merits for theoretical information theory study. For example, the close form capacity analysis could be proceeded. One recent work is \cite{Zhan1607:AchievableRates}, in which the lower bounds of the mutual information between the joint input and output\footnote{The joint input means the vector of eigenvalue and spectral amplitude.} were derived for soliton communication systems using the Gaussian approximated version \cite{SPAWC201506ZC} of the model in this paper.

Furthermore, our work is more general in the sense that assumptions useful for communication system design but unnecessary for theoretical model deriving were now dropped, for a more accurate analysis compared to related works (e.g. \cite{1986GH}). Specifically, unlike \cite{1986GH}, the perturbation of the imaginary part of the eigenvalue is taken into consideration as well in this paper for a better characterisation of the noise in a spectral amplitude. In addition, no Gaussian approximation is made on the noise distribution of the eigenvalue perturbation, which leads to a more precise estimation of the noise statistics of a spectral amplitude.


\section{Preliminaries}\label{Preli}

This section is about some background knowledge that is useful to this paper. We first introduce the concept of the NFT, and the $N$-soliton, the signal of interest in this paper, is mentioned. Then the transmission scheme using the NFT, named NFDM, is explained. 
Notations frequently used in this paper are given at last.

\subsection{The Nonlinear Fourier Transform}

The NFT, also called direct scattering transform, is a mathematical tool to solve integrable partial differential equations. It helps transfer a nonlinear evolution equation to a set of linear problems that are easier to solve. Unlike FT or Laplace transform which have generic forms, the NFT is indeed a series of procedures to obtain a representation in another domain, although it is called a transform.

In this subsection, we restrict our discussion to the deterministic NLSE
\begin{equation}\label{ch3.2.1}
{j\frac{\partial q(t,z)}{\partial z}=\frac{\partial^{2} q(t,z)}{\partial t^{2}}+2|q(t,z)|^{2}q(t,z).}
\end{equation}
To define the NFT of a signal, we need to find a pair of $q$-dependent operators $L$ and $M$, which satisfies
\begin{equation}\label{ch3.2.2}
{[M,L]=ML-LM}
\end{equation}
when \eqref{ch3.2.1} holds as a compatibility condition. When \eqref{ch3.2.2} holds, the eigenvalues of the operator $L$ are $z$-invariant as the signal $q(t,z)$ propagates through the fibre, and the operators $L$ and $M$ are called a Lax pair. For the NLSE \eqref{ch3.2.1}, we have
\begin{equation}\label{ch3.2.3}
{
L =
j\left( \begin{array}{ccc}
\frac{\partial}{\partial t} & -q(t,z) \\
-q^{*}(t,z) & -\frac{\partial}{\partial t}
\end{array} \right),
}
\end{equation}
and
\begin{equation}\label{ch3.2.4}
{M=
\left(\begin{array}{ccc}
2j\lambda^{2}-j|q(t,z)|^{2} & -2\lambda q(t,z)-jq_{t}(t,z)\\
2\lambda q^{*}(t,z)-jq_{t}^{*}(t,z) & -2j\lambda^{2}+j|q(t,z)|^{2}
\end{array} \right),}
\end{equation}
where $\lambda$ is an element in the spectrum of the operator $L$. In the following, we only consider the NFT for the input signal $q(t)\triangleq q(t,0)$, and suppress the variable $z$ because it is not related to the process of defining the NFT, and is only useful for showing the spatial signal propagation through an optical fibre. Throughout this paper, we assume that 
\begin{equation}\label{ch3.2.5}
{q(t)\in L^{1}(\mathbb{R}),}
\end{equation}
and
\begin{equation}\label{ch3.2.6}
{q(t)\rightarrow 0,\qquad t\rightarrow\infty.}
\end{equation}

The NFT of a signal $q(t)$ satisfying \eqref{ch3.2.5}--\eqref{ch3.2.6} is defined via the spectral analysis of the operator $L$. Specifically, this is done in the following steps.
\begin{enumerate}
{
\item Solve the eigenvalue problem $Lv=\lambda v$;

The space of the eigenvectors subject to an eigenvalue is a two-dimensional space. The eigenvalue problem is equivalent to
\begin{equation}\label{ch3.2.7}
{v_{t}=\left(\begin{array}{ccc}
-j\lambda & q(t)\\
-q^{*}(t) & j\lambda
\end{array} \right)v,}
\end{equation}
which is called a scattering problem. Specific to the following two boundary conditions,
\begin{equation}\label{ch3.2.8}
{\lim_{t\rightarrow\infty}\left|v^{(1)}(t,\lambda)-\left(\begin{array}{ccc}
0\\
1
\end{array} \right)e^{j\lambda t}\right|=0,}
\end{equation}
and
\begin{equation}\label{ch3.2.9}
{\lim_{t\rightarrow-\infty}\left|v^{(2)}(t,\lambda)-\left(\begin{array}{ccc}
1\\
0
\end{array} \right)e^{-j\lambda t}\right|=0,}
\end{equation}
two solutions of the equation \eqref{ch3.2.7}, as denoted by $v^{(1)}(t,\lambda)$ and $v^{(2)}(t,\lambda)$, are obtained, which are called canonical eigenvectors. It was shown that $\{v^{(1)}(t,\lambda),\widetilde{v}^{(1)}(t,\lambda^{*})\}$, and $\{v^{(2)}(t,\lambda),\widetilde{v}^{(2)}(t,\lambda^{*})\}$ are two sets of basis of the eigenvector space, where $\widetilde{v}^{(k)}(t,\lambda^{*})\triangleq\left(v^{(k)*}_{2}(t,\lambda),-v^{(k)*}_{1}(t,\lambda)\right)^{\mathrm{T}}$, and $v_{1}^{(k)}(t,\lambda)$ and $v_{2}^{(k)}(t,\lambda)$ are the first and second component of the canonical eigenvector $v^{(k)}(t,\lambda)$, $k=1,2$.

\item Obtain the scattering data;

As $\{v^{(1)}(t,\lambda),\widetilde{v}^{(1)}(t,\lambda^{*})\}$ and $\{v^{(2)}(t,\lambda),\widetilde{v}^{(2)}(t,\lambda^{*})\}$ are two sets of basis of the eigenvector space, we can project one set to the other, which is
\begin{equation}\label{ch3.2.10}
{v^{(2)}(t,\lambda)=a(\lambda)\widetilde{v}^{(1)}(t,\lambda^{*})+b(\lambda)v^{(1)}(t,\lambda),}
\end{equation}
and
\begin{equation}\label{ch3.2.11}
{\widetilde{v}^{(2)}(t,\lambda^{*})=b^{*}(\lambda^{*})\widetilde{v}^{(1)}(t,\lambda^{*})-a^{*}(\lambda^{*})v^{(1)}(t,\lambda).}
\end{equation}
The coefficients $a(\lambda)$ and $b(\lambda)$ are called scattering data, which can be obtained by calculating
\begin{equation}\label{ch3.2.12}
{a(\lambda)=\lim_{t\rightarrow\infty}v_{1}^{(2)}(t,\lambda)e^{j\lambda t},}
\end{equation}
and
\begin{equation}\label{ch3.2.13}
{b(\lambda)=\lim_{t\rightarrow\infty}v_{2}^{(2)}(t,\lambda)e^{-j\lambda t}.}
\end{equation}

\item The nonlinear Fourier transform.

The NFT of a signal $q(t)$ is composed by its spectrum and the corresponding spectral amplitudes. The spectrum of the operator $L$ is composed by the following two parts:
\begin{enumerate}
\item Discrete spectrum: The zeros of the scattering data $a(\lambda)$ on the upper half complex plane $\mathbb{C}^{+}\triangleq\{c\in\mathbb{C}:\ \mathrm{Im}(c)>0\}$. The elements of the discrete spectrum are called (discrete) eigenvalues;
\item Continuous spectrum: The whole real line $\mathbb{R}$.
\end{enumerate}
The spectral amplitudes also consist of two parts:
\begin{enumerate}
{
\item Discrete spectral amplitudes: The discrete spectral amplitude subject to an eigenvalue $\zeta_{k}\in\mathbb{C}^{+}$ is
\begin{equation}\label{ch3.2.14}
{Q^{(d)}(\zeta_{k})=\frac{b(\zeta_{k})}{a'(\zeta_{k})},\qquad k=1,2,\ldots,N,}
\end{equation}
where $a'(\zeta_{k})\triangleq \frac{\d a(\lambda)}{\d \lambda}\Big|_{\lambda=\zeta_{k}}$, and $N$ is the number of the zeros of $a(\lambda)$;
\item Continuous spectral amplitudes:
\begin{equation}\label{ch3.2.15}
{Q^{(c)}(\lambda)=\frac{b(\lambda)}{a(\lambda)},}
\end{equation}
where $\lambda\in\mathbb{R}$.
}
\end{enumerate}
}
\end{enumerate}

We already show that the spectrum of the signal keeps invariant as a signal propagates through an optical fibre in the noise free case. The spatial evolution of the spectral amplitudes are summarized as follows:
\begin{equation}\label{ch3.2.16}
{Q^{(c)}(\lambda,z)=Q^{(c)}(\lambda,0)e^{-4j\lambda^{2}z},}
\end{equation}
and
\begin{equation}\label{ch3.2.17}
{Q^{(d)}(\zeta_{k},z)=Q^{(d)}(\zeta_{k},0)e^{-4j\zeta_{k}^{2}z},\qquad k=1,2,\ldots,N,}
\end{equation}
where $Q^{(c)}(\lambda,z)$ and $Q^{(d)}(\zeta_{k},z)$ are respectively a continuous and a discrete spectral amplitude at position $z$, and $z>0$.

\subsection{$N$-Solitons}

We have already known that the NFT of a signal is composed by a continuous spectral function and discrete spectral functions. When there is no continuous spectral function, we obtain a special class of signals called $N$-solitons, whose definition are given as follows.

\begin{df}[N-solitons]\label{ch3.d2.1}
{The signals which do not have continuous spectrum, and only have $N$ eigenvalues forming its discrete spectrum are called $N$-solitons, where $N$ is a positive integer.}
\end{df}

When $N=1$, i.e., the spectrum of a signal is composed by one eigenvalue only, this signal is called a fundamental soliton, or abbreviated as ``soliton'' for convenience. This class of signals has fundamental importance. The shape of a soliton remains the same during the signal propagation because the dispersion effect, which is responsible for temporal broadening, and the Kerr nonlinearity, which is responsible for spectral broadening, are balanced. The analytical expression of a soliton with an eigenvalue $\zeta=\alpha+j\beta$ and a spectral amplitude $Q^{(d)}(\zeta,z)$ (noise free case) is
\begin{multline}
q(t,z)=2\beta\mathrm{sech}\left[2\beta \left(t-\frac{1}{2\beta}\ln \frac{|Q^{(d)}(\zeta,z)|}{2\beta}\right)\right]e^{-2j\alpha t-j\left(\arg Q^{(d)}(\zeta,z)+\frac{\pi}{2}\right)},\qquad t\in\mathbb{R},\ \mathrm{and}\ z\geq 0.\label{ch3.3.1}
\end{multline}

\subsection{The Nonlinear Frequency Division Multiplexing Scheme}

Because of the chromatic dispersion and the Kerr nonlinearity, signals from different users couple together in a complicated manner, and interference between users is inevitable. There are multiplexing schemes proposed for multiuser communications, such as time division multiplexing (TDM), orthogonal frequency division multiplexing (OFDM), and wavelength division multiplexing (WDM), however, they do not work well.

In the recent work \cite{DBLP:journals/corr/abs-1202-3653}, Yousefi and Kschischang propose a new multiplexing technique, in which the NFT is applied as the main tool, such that interference free communication is theoretically possible at least in the noise free case. In fact, to some extent, the idea is the same as the OFDM in which Fourier transform (FT) is invoked to transform a linear time invariant channel into a set of independent parallel channels (in the frequency domain). Specifically, with the help of the NFT, the time domain nonlinear dispersive channel is transferred to a set of linear multiplicative channels in the (nonlinear) spectral domain \eqref{ch3.2.16}--\eqref{ch3.2.17}. Furthermore, we notice from \eqref{ch3.2.17} that the spatial evolutions of the discrete spectral amplitudes corresponding to their eigenvalues are independent of each other in the noise free case. This property leads us a channel diagonalisation in the spectral domain, which allows us to eliminate the interference between users if the spectrum band is properly allocated.

The process of information transmission using the nonlinear Fourier transform (NFT) in the noise free case \cite{DBLP:journals/corr/abs-1202-3653} is summarised as follows (see also Figure \ref{ch3.f1.1}). The channel input is encoded and modulated in the spectral domain (denoted by $Q(\lambda,0)$) at first. Applying the INFT, the time domain input signal $q(t,0)$ is obtained. At the receiver, the NFT is used to transfer the time domain received signal $q(t,\length)$ to its spectral domain representation $Q(\lambda,\length)$, where $\length$ is the length of the fibre. The decoding is done in the spectral domain using the spectral domain signal input-and-output relationship \eqref{ch3.2.16}--\eqref{ch3.2.17}. This transmission scheme is called nonlinear frequency division multiplexing (NFDM).

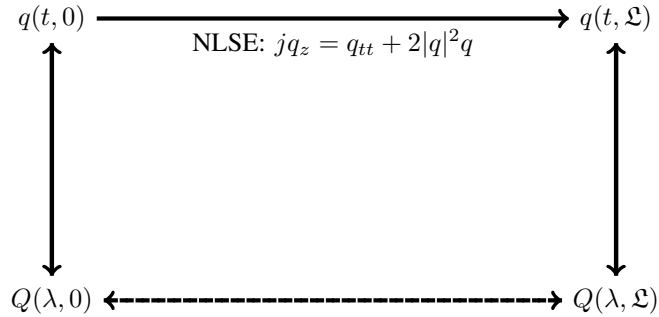
\begin{figure}[!htb]
\centering
\begin{tikzpicture}
[->,shorten >=0pt,auto,node distance=3.75cm, line width=1.6pt]
  \tikzstyle{every state}=[rectangle,fill=none,draw=none,text=black,minimum size=6mm]

  \node[state]         (i)                 { $Q(\lambda,0)$};
  \node[state]         (j)[above of=i]   { $q(t,0)$};
  \node[state]         (k)[right of=j]   {};
  \node[state]         (m)[right of=k]   { $q(t,\mathfrak{L})$};
  \node[state]         (l)[below of=m]   { $Q(\lambda,\mathfrak{L})$};
\path(i) edge           node[swap] {} (j);
\path(j) edge           node[swap] {} (i);
\path(j) edge           node[swap] { NLSE:\ $jq_{z}=q_{tt}+2|q|^{2}q$} (m);
\path(m) edge           node[swap] {} (l);
\path(l) edge           node[swap] {} (m);
\path(i) edge[dashed]   node[swap] {} (l);
\path(l) edge[dashed]   node[swap] {} (i);
\end{tikzpicture}
\caption{Information transmission using the nonlinear Fourier transform.}\label{ch3.f1.1}
\end{figure}

\subsection{Soliton Perturbation Theory}

In \cite{b25}, there are results about how are the scattering parameters of soliton solutions perturbed due to the fluctuation of the original NLSE \eqref{ch3.2.1} by a small deterministic perturbation ($j\epsilon G(t,z)$). Note that they are only valid for solitons, and are only precise to the first order. The first order perturbation of the real part of the eigenvalue of a soliton is
\begin{equation}
\frac{\partial \alpha(z)}{\partial z}=-\epsilon\int_{-\infty}^{\infty}\im\left[G(t,z)e^{-j\varphi(t,z)}\right]\beta(z)\mathrm{sech}[2\beta(z)(t-T_{0}(z))]\mathrm{tanh}[2\beta(z)(t-T_{0}(z))]\d t,\label{ch3.3.6}
\end{equation}
and that of the imaginary part of the eigenvalue is
\begin{equation}
\frac{\partial\beta(z)}{\partial z}=\epsilon\int_{-\infty}^{\infty}\re\left[G(t,z)e^{-j\varphi(t,z)}\right]\beta(z)\mathrm{sech}[2\beta(z)(t-T_{0}(z))]\d t,\label{ch3.3.2}
\end{equation}
and that of the centre of the soliton pulse $T_{0}(z)$ is
\begin{equation}
\frac{\partial T_{0}(z)}{\partial z}=4\alpha(z)+\epsilon\int_{-\infty}^{\infty}\re\left[G(t,z)e^{-j\varphi(t,z)}\right](t-T_{0}(z))\mathrm{sech}[2\beta(z)(t-T_{0}(z))]\d t,\label{ch3.3.32}
\end{equation}
and that of the soliton parameter $\theta(z)=-2\alpha(z)T_{0}(z)-\arg Q^{(d)}(\zeta(z),z)-\frac{\pi}{2}$ involving the phase of a spectral amplitude is
\begin{multline}
\frac{\partial\theta(z)}{\partial z}=-4\left[\alpha(z)^{2}+\beta(z)^{2}\right]+\epsilon\int_{-\infty}^{\infty}\im\left[G(t,z)e^{-j\varphi(t,z)}\right]\mathrm{sech}\left[2\beta(z)(t-T_{0}(z))\right]\\
\{1-2\beta(z)(t-T_{0}(z))\mathrm{tanh}\left[2\beta(z)(t-T_{0}(z))\right]\}\d t\\
+2\epsilon\alpha(z)\int_{-\infty}^{\infty}\re\left[G(t,z)e^{-j\varphi(t,z)}\right](t-T_{0}(z))\mathrm{sech}[2\beta(z)(t-T_{0}(z))]\d t,\label{ch3.3.33}
\end{multline}
where 
\begin{equation}\label{ch2.4.1}
T_{0}(z)=\frac{1}{2\beta(z)}\ln \frac{|Q^{(d)}(\zeta(z),z)|}{2\beta(z)},
\end{equation}
and
$$\varphi(t,z)=-2\alpha(z)t-\arg Q^{(d)}(\zeta(z),z)-\frac{\pi}{2}.$$


When considering the noise model in the information theoretic point of view, the input and the noise need to be treated as random variables (or random processes). However, in the deterministic perturbation results \eqref{ch3.3.6}--\eqref{ch3.3.33}, the noise term $G(t,z)$ in \eqref{ch3.1.1} was treated as a deterministic perturbation, and hence, need to be modified to be stochastic.

The stochastic perturbation results for soliton parameters were derived in \cite{2005DTY}. Since the stochastic calculus in Stratonovich sense adopts the same symbolic rules as ordinary calculus does to derive the deterministic perturbation theory, equations \eqref{ch3.3.6}--\eqref{ch3.3.33} also give the stochastic perturbation theory in Stratonovich  sense, i.e., the integrals therein are all Stratonovich stochastic integrals.

The stochastic perturbation results in It\^{o} sense are much easier to work with compared to those in Stratonovich sense. For this reason, they were transferred to It\^{o} sense in \cite{2005DTY}. Mathematically, stochastic integrals in both these two senses are equally acceptable, however, they look different by an advection term because the rules of the It\^{o} calculus are different from those of the other. The relationship between the integrals in those two senses, in particular the derivation of the advection term, can be found in classical textbooks \cite{b42,b33}.

In this paper, we adopt the It\^{o} stochastic calculus.  The It\^{o} stochastic perturbation results \cite{2005DTY} are summarised as follows. The first order It\^{o} stochastic perturbations of the real part of the eigenvalue, the centre of the soliton pulse, and the parameter $\theta(z)$ involving the phase of its spectral amplitude are respectively of the same forms as those in \eqref{ch3.3.6}, \eqref{ch3.3.32}, and \eqref{ch3.3.33}, i.e., the advection terms are all zero. The first order It\^{o} stochastic perturbations of the imaginary part of the eigenvalue is
\begin{equation}
\frac{\partial\beta(z)}{\partial z}=\frac{1}{2}\epsilon^{2}+\epsilon\int_{-\infty}^{\infty}\re\left[G(t,z)e^{-j\varphi(t,z)}\right]\beta(z)\mathrm{sech}[2\beta(z)(t-T_{0}(z))]\d t,\label{ch3.3.2s}
\end{equation}
where $\frac{1}{2}\epsilon^{2}$ is the advection term, and the integral on the right hand side of \eqref{ch3.3.2s} is an It\^{o} stochastic integral.


\subsection*{Notations}

We summarize the notations used in this paper in this subsection. Denote $\length>0$ the length of the optical fibre. For an input signal $q(t,z)$, we denote its discrete eigenvalues by $\zeta_{i}(z)=\alpha_{i}(z)+j\beta_{i}(z)$, where $\alpha_{i}(z), \beta_{i}(z)\in\mathbb{R}$, and $\beta_{i}(z)>0$, $i=1,2,\ldots,N$. Its spectral amplitudes are denoted by $Q^{(d)}(\zeta_{i}(z),z)$,  where $z\in [0,\length]$. Note that in the noise free case, $\zeta_{i}(z)$ is a constant for any $z\in [0,\length]$, and hence is denoted by $\zeta_{i}\triangleq\alpha_{i}+j\beta_{i}$ for convenience. 

Denote respectively
\begin{equation}\label{ch3.3.52}
{\nri(z)\triangleq \alpha_{i}(z)-\alpha_{i}(0),\qquad\forall\ z\in[0,\length],}
\end{equation} 
and
\begin{equation}\label{ch3.3.51}
{\nii(z)\triangleq \beta_{i}(z)-\beta_{i}(0),\qquad\forall\ z\in[0,\length]}
\end{equation}
the noise processes for the real and imaginary parts of the eigenvalue $\zeta_{i}(z)$ of an input signal, $i=1,2,\ldots,N$. We further denote $\nid(z)\triangleq\nri(z)+j\nii(z)$, then $\nid(z)$ is the (non-Gaussian) noise process for the eigenvalue, $z\in[0,\length]$, $i=1,2,\ldots,N$. For solitons, we drop the subscripts in the notations in \eqref{ch3.3.52} and \eqref{ch3.3.51}, i.e., using $\n(z)$, $\nr(z)$ and $\ni(z)$ to denote respectively the (non-Gaussian) noise processes for the eigenvalue and its real and the imaginary parts. We further denote
\begin{equation}\label{ch3.3.118}
{\nu_{I}(z)\triangleq\ni(z)-\frac{1}{2}\epsilon^{2}z,\qquad\forall\ z\in[0,\length],}
\end{equation}

\section{A Channel Model for Discrete Spectral Amplitudes}\label{Model}

In this section, we make the first effort in proposing a channel model for the noisy evolution of discrete spectral amplitudes. Noise modelling is an interesting and important topic in digital communications \cite{b18} in terms of leading to the fundamental limit for reliable communications, namely channel capacity \cite{b13}. We first introduce our modelling methodology, and then propose the channel model. The signals of interest in this section are restricted to $N$-solitons. We will discuss later that the modelling methodology might also work for the discrete spectral amplitudes of the inputs with continuous spectrum as well.


Throughout this paper, we use an underline under a variable to denote a random variable when it is necessary.

\subsection{Modelling Methodology}



According to the system model \eqref{ch3.1.1}, we assume that distributed amplification is used, and hence the   noise is injected during the whole transmission process. 
To characterise the noise, 
a typical approach to model the signal propagation in a noisy channel is used, which is to divide the fibre into many extremely short segments so that signal propagation along each segment is modelled by two phases. In the first phase, an additive noise will be added. In the second phase, noise will be ignored and signal will be propagated along a noiseless fibre. 
Under this model, during the first phase, the added noise may cause perturbation in the discrete eigenvalues $\zeta_{i}$ and also the spectral amplitudes of $N$-solitons $Q^{(d)}(\zeta_{i})$, $i=1,2,\ldots,N$. In the second phase, the perturbation in each eigenvalue will incur additional noise to its corresponding spectral amplitude, as in $\zeta_{i}$ in the term $e^{-4j\zeta_{i}^{2}z}$ in \eqref{ch3.2.17}, the noiseless spectral domain signal evolution. Then by concatenating all these segments together, we can characterise the aggregate noises added during the transmission.

In this paper, our model focuses primarily on the noise caused by eigenvalue perturbation, which is significant as shown in Section \ref{Discussions} for scenarios of interest, such as long distance communication and high input power. This is also supported by the analysis in the celebrated Gordon-Haus effect \cite{1986GH}: In a special case when $N=1$, time jittering of a soliton was claimed to be largely due to the perturbation of the real part of the eigenvalue (or group velocity) \cite{1986GH}. On 
the other hand, the centre of a soliton is closely related to its magnitude of the spectral amplitude. Therefore, this motivates the idea that the perturbation in the eigenvalue may have a significant impact on the spectral amplitude of a soliton as well.  
Notice that the noise in the spectral amplitude incurred in the second phase depends on the soliton's eigenvalue. To some extent,  our model can be viewed as the special case when the noise in a spectral amplitude is dominated by those incurred in the second phase.


To further illustrate the idea, suppose we split an optical fibre of a fixed length $\length$ into $m$ segments, $m\in\mathbb{Z}^{+}$. Specifically, let 
\begin{equation}\label{ch3.3.35}
{0\triangleq c_{0}<c_{1}<\cdots<c_{m-1}<c_{m}\triangleq\length}.
\end{equation}
Then the  $k$-th fibre piece or segment (denoted by $I_{k}$) corresponds to the piece at the interval $ [c_{k-1},c_{k})$. The length of the $k$-th segment, denoted by $|I_{k}|$, thus equals to $ c_{k}-c_{k-1}$. 
Here, we assume that $m$ is large, and that $|I_{k}|$ will be vanishingly small, $k=1,2,\ldots,m$.

Let us now focus on a particular segment $I_{k}$ (corresponding to the interval $ [c_{k-1},c_{k})$). At the input of the $k$-th segment, suppose the eigenvalues are $\zeta_{i}(c_{k-1})$, with respectively their corresponding spectral amplitudes $Q^{(d)}(\zeta_{i}(c_{k-1}) , c_{k-1})$, $i=1,2,\ldots,N$.
We assume that noise is injected at the beginning of the segment $c_{k-1}$. Since the length of the segment is small, the noise added at  this point will  be very small. 
%

At the output of the $k$-th segment, the eigenvalues $\zeta_{i}(c_{k})$ 
will be perturbed by the time domain white Gaussian noise process (the term $G(t,z)$ in the SNLSE \eqref{ch3.1.1}). 
 The changes in each eigenvalue will also affect the changes (or evolution) of the corresponding spectral amplitude. Specifically, we model that for $i=1,2,\ldots,N$,
\begin{equation}\label{ch3.3.53}
{Q^{(d)}_{m}(\zeta_{i}(c_{k}),c_{k})\triangleq Q^{(d)}_{m}(\zeta_{i}(c_{k-1}),c_{k-1})e^{-4j\zeta_{i}(z_{k})^{2}|I_{k}|}},
\end{equation}
where $z_{k}$ is some arbitrary point in the interval $[c_{k-1},c_{k})$, and $\zeta_{i}(z_{k})$ is the eigenvalue at $z=z_{k}$, which is the same as the eigenvalue after being perturbed at the beginning of this fibre segment because of the noiseless evolution in the second phase of noise modelling. Here, we  use $Q^{(d)}_{m}(\zeta_{i}(c_{k}),c_{k})$ instead of 
$Q^{(d)}(\zeta_{i}(c_{k}),c_{k})$ to emphasise that it is an approximation, and that its accuracy  will depend on the size of $m$.

%
%
%
%

%

Now, by concatenating all the fibre segments,  we have for $i=1,2,\ldots,N$,
\begin{align}
Q^{(d)}(\zeta_{i}(\length),\length) 
& \approx Q^{(d)}_{m}(\zeta_{i}(\length),\length) \\
& = Q^{(d)}(\zeta_{i}(0),0) 
\prod_{k=1}^{m}e^{-4j\zeta_{i}(z_{k})^{2}|I_{k}|} \\
& = Q^{(d)}(\zeta_{i}(0),0) 
\prod_{k=1}^{m}e^{-4j\left[\zeta_{i}(0)+\nid(z_{k})\right]^{2}|I_{k}|}. \label{ch3.3.38}
\end{align}
%
%
%
Finally,  by  studying the asymptotic behaviour of \eqref{ch3.3.38} (as $m$ goes to infinity, which means that the lengths of the fibre segments tend to zero), we obtain the channel model for the spectral amplitudes.

For the spectral amplitude $Q^{(d)}(\zeta_{i},z)$, we regard $Q^{(d)}(\zeta_{i}(0),\length)$, a scaled input of the actual input spectral amplitude $Q^{(d)}(\zeta_{i}(0),0)$, as the channel input of a spectral amplitude for the purpose of convenience. Our idea is that we do not regard the noiseless spatial evolution of a spectral amplitude
\begin{equation}\label{ch3.x0.1}
{Q^{(d)}(\zeta_{i}(0),\length)=e^{-4j\zeta_{i}(0)^{2}\length}Q^{(d)}(\zeta_{i}(0),0)}
\end{equation}
as part of the noise, and only compare the difference between the scaled input $Q^{(d)}(\zeta_{i}(0),\length)$ and output $Q^{(d)}(\zeta_{i}(\length),\length)$, which is called the noise in a spectral amplitude, $i=1,2,\ldots,N$.


\subsection{A Noisy Input-and-output Relationship of Discrete Spectral Amplitudes}

In this subsection, we derive an analytical channel model for the discrete spectral amplitudes of an $N$-soliton.

\begin{thm}\label{ch3.t3.9}
{If the stochastic processes $\nri(z)$, $\nii(z)$ and $\nri(z)\nii(z)$ are all mean-square integrable conditioned on any input eigenvalue $\zeta_{i}(0)$ over the interval $[0,\length]$, then the channel model for the magnitudes of the spectral amplitudes of an $N$-soliton (under the proposed model) is given by
\begin{multline}
\ln |Q^{(d)}(\zeta_{i}(\length),\length)|=\ln |Q^{(d)}(\zeta_{i}(0),\length)|+8\alpha_{i}(0)\int_{0}^{\length}\nii(z)\d z+8\beta_{i}(0)\int_{0}^{\length}\nri(z)\d z\\
+8\int_{0}^{\length}\nri(z)\nii(z)\d z,\label{ch3.3.46n}
\end{multline}
where $Q^{(d)}(\zeta_{i}(0),\length)$ is defined in \eqref{ch3.x0.1},
$i=1,2,\ldots,N$. Here, the integrals in \eqref{ch3.3.46n} are all mean-square stochastic integrals.}
\end{thm}

\begin{proof}[Proof of Theorem \ref{ch3.t3.9}]
{
We split an optical fibre of length $\length$ into $m$ pieces by making an arbitrary partition of the interval $[0,\length]$ given in \eqref{ch3.3.35},
where $m$ is a positive integer. Denote the $k$-th fibre piece by $I_{k}\triangleq[c_{k-1},c_{k})$, and denote its length by $|I_{k}|\triangleq c_{k}-c_{k-1}$, $k=1,2,\ldots,m$. Then for any $k\in\{1,2,\ldots,m\}$, we choose a $z_{k}\in[c_{k-1},c_{k})$ such that \eqref{ch3.3.53} holds. So for $i=1,2,\ldots,N$, the noisy evolution of the spectral amplitude when we split the fibre into $m$ pieces as described in \eqref{ch3.3.35} is \eqref{ch3.3.38}, which gives us 
\begin{eqnarray}
\ln Q^{(d)}_{m}(\zeta_{i}(\length),\length) &=& \ln Q^{(d)}(\zeta_{i}(0),0)-4j\zeta_{i}(0)^{2}\length-8j\zeta_{i}(0)\sum_{k=1}^{m}\nid(z_{k})|I_{k}|-4j\sum_{k=1}^{m}\nid(z_{k})^{2}|I_{k}|+2js_{i}\pi\nonumber\\
&=&\ln Q^{(d)}(\zeta_{i}(0),\length)-8j\zeta_{i}(0)\sum_{k=1}^{m}\nid(z_{k})|I_{k}|-4j\sum_{k=1}^{m}\nid(z_{k})^{2}|I_{k}|+2js_{i}\pi\label{ch3.3.39},
\end{eqnarray}
where $z_{k}\in[c_{k-1},c_{k})$, $k=1,2,\ldots,m$, and $s_{i}$ is a random variable taking integer values. Denote $N_{i,m}(\length)\triangleq\ln |Q^{(d)}_{m}(\zeta_{i}(\length),\length)|-\ln |Q^{(d)}(\zeta_{i}(0),\length)|$, then according to \eqref{ch3.3.39}, its real part is
\begin{equation}
N_{i,m}(\length)=8\alpha_{i}(0)\sum_{k=1}^{m}\nii(z_{k})|I_{k}|+8\beta_{i}(0)\sum_{k=1}^{m}\nri(z_{k})|I_{k}|+8\sum_{k=1}^{m}\nri(z_{k})\nii(z_{k})|I_{k}|\label{ch3.3.40}.
\end{equation}

Since the stochastic processes $\nri(z)$, $\nii(z)$ and $\nri(z)\nii(z)$ are all assumed to be mean-square integrable over the interval $[0,\length]$, $i=1,2,\ldots,N$, we have
\begin{equation}\label{ch3.3.42}
{\mslim_{\Delta\rightarrow 0}\sum_{k=1}^{m}\nri(z_{k})|I_{k}|=\int_{0}^{\length}\nri(z)\d z,}
\end{equation}
and
\begin{equation}\label{ch3.3.43}
{\mslim_{\Delta\rightarrow 0}\sum_{k=1}^{m}\nii(z_{k})|I_{k}|=\int_{0}^{\length}\nii(z)\d z,}
\end{equation}
and
\begin{equation}\label{ch3.3.44}
{\mslim_{\Delta\rightarrow 0}\sum_{k=1}^{m}\nri(z_{k})\nii(z_{k})|I_{k}|=\int_{0}^{\length}\nri(z)\nii(z)\d z,}
\end{equation}
where we use ``$\mslim$'' to denote the mean-square limit, and $\Delta\triangleq\max_{k=1,2,\ldots,m}|I_{k}|$.

Since the noise processes $\nri(z)$ and $\nii(z)$ are assumed to be mean-square integrable in $[0,\length]$ conditioned on a particular choice of input eigenvalue $\zeta_{i}(0)=\alpha_{i}(0)+j\beta_{i}(0)$, $i=1,2,\ldots,N$. According to Lemma \ref{ch3.l3.10}, we have 
\begin{equation}\label{ch3.3.88}
{\mathrm{E}\left\{\left[\underline{\alpha_{i}(0)}\sum_{k=1}^{m}\underline{\nii(z_{k})}|I_{k}|-\underline{\alpha_{i}(0)}\int_{0}^{\length}\underline{\nii(z)}\d z\right]^{2}\Bigg|\underline{\zeta_{i}(0)}=\zeta_{i}(0)\right\}\rightarrow 0,\qquad \Delta\rightarrow 0,}
\end{equation}
and
\begin{equation}\label{ch3.3.89}
{\mathrm{E}\left\{\left[\underline{\beta_{i}(0)}\sum_{k=1}^{m}\underline{\nri(z_{k})}|I_{k}|-\underline{\beta_{i}(0)}\int_{0}^{\length}\underline{\nri(z)}\d z\right]^{2}\Bigg|\underline{\zeta_{i}(0)}=\zeta_{i}(0)\right\}\rightarrow 0,\qquad \Delta\rightarrow 0.}
\end{equation}

So according to the Law of total expectation, we have 
\begin{equation}\label{ch3.3.90}
{\mathrm{E}\left[\underline{\alpha_{i}(0)}\sum_{k=1}^{m}\underline{\nii(z_{k})}|I_{k}|-\underline{\alpha_{i}(0)}\int_{0}^{\length}\underline{\nii(z)}\d z\right]^{2}\rightarrow 0,\qquad\Delta\rightarrow 0.}
\end{equation}
and
\begin{equation}\label{ch3.3.91}
{\mathrm{E}\left[\underline{\beta_{i}(0)}\sum_{k=1}^{m}\underline{\nri(z_{k})}|I_{k}|-\underline{\beta_{i}(0)}\int_{0}^{\length}\underline{\nri(z)}\d z\right]^{2}\rightarrow 0,\qquad\Delta\rightarrow 0.}
\end{equation}
which respectively mean that
\begin{equation}\label{ch3.3.99}
{\mslim_{\Delta\rightarrow 0}\underline{\alpha_{i}(0)}\sum_{k=1}^{m}\underline{\nii(z_{k})}|I_{k}|=\underline{\alpha_{i}(0)}\int_{0}^{\length}\underline{\nii(z)}\d z,}
\end{equation}
and
\begin{equation}\label{ch3.3.100}
{\mslim_{\Delta\rightarrow 0}\underline{\beta_{i}(0)}\sum_{k=1}^{m}\underline{\nri(z_{k})}|I_{k}|=\underline{\beta_{i}(0)}\int_{0}^{\length}\underline{\nri(z)}\d z.}
\end{equation}

Applying Lemma \ref{ch3.l3.10} and equations \eqref{ch3.3.44}, \eqref{ch3.3.99}--\eqref{ch3.3.100}, the mean-square limit of the right hand side of \eqref{ch3.3.40} will be the noise in the magnitude of the spectral amplitude, which is
\begin{equation}
N_{i}(\length)\triangleq\mslim_{\Delta\rightarrow 0}N_{i,m}(\length)=8\alpha_{i}(0)\int_{0}^{\length}\nii(z)\d z+8\beta_{i}(0)\int_{0}^{\length}\nri(z)\d z+8\int_{0}^{\length}\nri(z)\nii(z)\d z.\label{ch3.3.45}
\end{equation}
for $i=1,2,\ldots,N$. Under our modelling methodology, we have
\begin{equation}\label{ch3.3.101}
{N_{i}(\length)\triangleq\ln|Q^{(d)}(\zeta_{i}(\length),\length)|-\ln|Q^{(d)}(\zeta_{i}(0),\length)|.}
\end{equation}
So the noise model of the magnitudes of spectral amplitudes is \eqref{ch3.3.46n}.
}
\end{proof}

We note that the noise in each discrete spectral amplitude $Q^{(d)}(\zeta_{i}(\length),\length)$ has only the noise incurred in the corresponding eigenvalue $\nid(z)$ involved, but not the others. To see the reason why the perturbations of other eigenvalues are not involved, we recall our modelling method. For each fibre segment, spectral domain signal parameters are perturbed at first (the first phase), and then propagated noiselessly with only the perturbation of the eigenvalue in the first phase incurring additional contamination to the spectral amplitude during noiseless evolution (the second phase).  Since the noiseless spatial evolution of a spectral amplitude (in second phase) only depends on its corresponding eigenvalue (see \eqref{ch3.2.17}), the additional noise incurred to the spectral amplitude in this phase only depends on the perturbation of the corresponding eigenvalue, although the perturbation of the $N$ eigenvalues may statistically correlated with each others. 
To be more specific, we notice that the eigenvalue $\zeta_{k}$ contributes to the spatial evolution of the spectral amplitude in terms of the term $e^{-4j\zeta_{k}^{2}z}$. Although the noise (denoted by $\n_{k}(z)$) contaminating the eigenvalue $\zeta_{k}(0)$ may not be independent of another $\n_{m}(z)$ statistically $(m\neq k)$, we could still obtain the noisy input-and-output relationship by $e^{-4j[\zeta_{k}(0)+\n_{k}(z)]^{2}z}$ for each soliton component $(\zeta_{k}, Q^{(d)}(\zeta_{k},0))$, $k=1,2,\ldots,N$.

For the same reason, the channel model we propose might also work for discrete spectral amplitudes of general signals, which have continuous spectrum as well, such as rectangular and sinc pulses. Although the noise in continuous spectrum might be statistically correlated with that in discrete eigenvalues, we could still get some ideas of how discrete spectral amplitudes are contaminated provided that the noise effects on discrete eigenvalues are a perturbation around their original ones at the input, i.e., neither new eigenvalue is generated, nor any eigenvalue vanishes.

Theorem \ref{ch3.t3.9} provides a noisy input-output relationship of the magnitudes of the spectral amplitudes. Their statistics are not available in close forms for general $N$-solitons when $N>1$ because the perturbation results of their eigenvalues are required in the first place, which are not available in simple forms.


We need to remark here that our channel model only characterises part of the noise actually incurred during the noisy evolution of discrete spectral amplitudes, i.e., the accumulation of the perturbation of the eigenvalue in the spatial evolution of the spectral amplitude (the second phase) is modelled, but not the initial perturbation of the spectral amplitude (the first phase). As a result, our model is an approximation. To the best of our knowledge, the discrete spectral amplitude model is essentially missing in the literature. While deriving a full and complete channel model is challenging in general, we target at the derivation of a simpler model which is easier and more suitable for theoretical analysis. In the following section, we consider a special case of our channel model with soliton inputs, where more analytical tools are accessible such as perturbation theory. We discuss in this scenario when the noise we capture in our proposed model is significant, or even dominant.

Similarly, if the stochastic processes $\nri(z)^{2}-\nii(z)^{2}$ are mean-square integrable over the interval $[0,\length]$, $i=1,2,\ldots,N$, the noise model for the phases of the spectral amplitudes of an $N$-soliton (under the proposed model) is
\begin{multline}
\arg Q^{(d)}(\zeta_{i}(\length),\length)\equiv\arg Q^{(d)}(\zeta_{i}(0),\length)-8\alpha_{i}(0)\int_{0}^{\length}\nri(z)\d z+8\beta_{i}(0)\int_{0}^{\length}\nii(z)\d z\\
-4\int_{0}^{\length}\left[\nri(z)^{2}-\nii(z)^{2}\right]\d z\qquad(\mathrm{mod}\ 2\pi)\label{ch3.4.102}
\end{multline}
for $i=1,2,\ldots,N$. However, the noise statistics are also unknown.

\section{A Noise Model for A Discrete Spectral Amplitude of A Soliton: A Special Case}\label{Special}

In the previous section, we propose a modelling methodology for the noise in discrete spectral amplitudes of $N$-solitons, and a channel model is derived using this method. We notice that the noise statistics are missing because the perturbation theory for the eigenvalues of $N$-solitons is required, which is not available. As a result, we consider a special case in this section, where the input signals are solitons. As is shown below, the conditions of the convergence of stochastic integrals could be derived in the channel model, and the noise statistics could be obtained as well.

\subsection{Noisy Input-and-output Relationship of the Magnitude of A Spectral Amplitude}


In this subsection, we develop a noise model for the magnitude of a spectral amplitude using the modelling methodology described in the previous section.

We first study some statistical properties of the noise terms that are necessary to derive the noise model.

\begin{lemma}\label{ch3.l3.9}
{Assume $0\leq s\leq t\leq \length$, then we have
\begin{enumerate}
{
\item 
\begin{equation}\label{ch3.3.57}
{\mathrm{E}\left\{\nr(s)^{2}\nu_{I}(s)\left[\nu_{I}(t)-\nu_{I}(s)\right]\right\}=0,}
\end{equation}
and
\begin{equation}\label{ch3.3.58}
{\mathrm{E}\left\{\nr(s)\nu_{I}(s)^{2}\left[\nr(t)-\nr(s)\right]\right\}=0,}
\end{equation}
and
\begin{equation}
\mathrm{E}\left\{\nr(s)\nu_{I}(s)\left[\nr(t)-\nr(s)\right]\left[\nu_{I}(t)-\nu_{I}(s)\right]\right\}=0;\label{ch3.3.59}
\end{equation}
\item 
\begin{multline}
\mathrm{E}\left[\nr(s)\nu_{I}(s)\nr(t)\nu_{I}(t)\right]=\mathrm{E}\left[\nr(s)^{2}\nu_{I}(s)^{2}\right]\\
=\frac{1}{12}\epsilon^{4}s^{2}\mathrm{E}\left[\beta(0)\right]^{2}+\frac{1}{18}\epsilon^{6}s^{3}\mathrm{E}\beta(0)+\frac{1}{144}\epsilon^{8}s^{4};\label{ch3.3.60}
\end{multline}
\item 
\begin{equation}\label{ch3.3.105}
{\mathrm{E}\left[\nr(s)^{2}\nu_{I}(s)\Big|\underline{\zeta(0)}=\zeta(0)\right]=0;}
\end{equation}
\item 
\begin{equation}
\mathrm{E}\left[\nr(s)\nr(t)\Big|\underline{\zeta(0)}=\zeta(0)\right]=\mathrm{E}\left[\nr(s)^{2}\Big|\underline{\zeta(0)}=\zeta(0)\right],\label{ch3.3.121}
\end{equation}
and
\begin{equation}
\mathrm{E}\left[\ni(s)\ni(t)\Big|\underline{\zeta(0)}=\zeta(0)\right]=\mathrm{E}\left[\ni(s)^{2}\Big|\underline{\zeta(0)}=\zeta(0)\right]+\frac{1}{4}\epsilon^{4}s(t-s).\label{ch3.3.122}
\end{equation}
Furthermore, the noise processes $\nr(z)$ and $\ni(z)$ are both mean-square integrable conditioned on an eigenvalue $\zeta(0)=\alpha(0)+\beta(0)$ over the interval $[0,\length]$;
\item If $\mathrm{E}\beta(0)<\infty$, and $\mathrm{E}[\beta(0)]^{2}<\infty$, the stochastic process $\nr(z)\ni(z)$ is mean-square integrable over the interval $[0,\length]$.
}
\end{enumerate}
}
\end{lemma}

\begin{proof}[Sketch of the Proof]
{The technique of the proof is essentially similar to that of proving Theorem \ref{ch3.t3.5}. Please refer to Appendix \ref{AppendixB} for a detailed proof.}
\end{proof}

With the help of Lemma \ref{ch3.l3.9}, we have the following channel model for the magnitude of a spectral amplitude.

\begin{thm}\label{ch3.t3.7}
{
If $\mathrm{E}\beta(0)<\infty$, and $\mathrm{E}[\beta(0)]^{2}<\infty$, the noise model for the magnitude of the spectral amplitude of a soliton (under the proposed modelling methodology) is given by
\begin{multline}
\ln |Q^{(d)}(\zeta(\length),\length)|=\ln |Q^{(d)}(\zeta(0),\length)|+8\alpha(0)\int_{0}^{\length}\ni(z)\d z+8\beta(0)\int_{0}^{\length}\nr(z)\d z\\
+8\int_{0}^{\length}\nr(z)\ni(z)\d z, \label{ch3.3.46}
\end{multline}
where
\begin{equation}\label{ch3.x3.1}
Q^{(d)}(\zeta(0),\length)= Q^{(d)}(\zeta(0),0)e^{-4j\zeta(0)^{2}\length},
\end{equation}
and $\nr(z)$ and $\ni(z)$ are the noise processes of the real and imaginary parts of the eigenvalue, $\alpha(z)$ and $\beta(z)$, respectively.
}
\end{thm}

\begin{proof}[Proof of Theorem \ref{ch3.t3.7}]
{The proof is essentially similar to that of Theorem \ref{ch3.t3.9} except that the mean-square convergence of the stochastic integrals in \eqref{ch3.3.42}--\eqref{ch3.3.100} can be proved with the help of Lemma \ref{ch3.l3.9}.
}
\end{proof}

\subsection{Noise Statistics of the Magnitude Channel}

We notice that existing results of the statistics of the noise in an eigenvalue, such as \cite{DBLP:journals/corr/abs-1302-2875}, often adopted an approximation of conditional Gaussian distribution of the noise in an eigenvalue, which is based on the assumption that the soliton parameters on the right hand sides of \eqref{ch3.3.6} and \eqref{ch3.3.2} remain unchanged. However, this assumption holds only when the fibre is extremely short. In this paper, we impose no such  assumption, and study non-Gaussian statistics of the noise in the magnitude of the spectral amplitude. The statistics obtained are more precise than the ones derived with Gaussian approximation. 

Throughout this paper, we use the first order perturbation theory \eqref{ch3.3.6} and \eqref{ch3.3.2} to derive noise statistics. As a result, the noise statistics are accurate only in the sense of the first order perturbation.

\subsubsection{Statistical Properties of the Noise in an Eigenvalue}


The non-Gaussian statistics of the noise in an eigenvalue were obtained in \cite{2005DTY}. They can also be calculated directly from the It\^{o} stochastic perturbation theory \eqref{ch3.3.6}--\eqref{ch3.3.2}, and are summarized as follows: Conditioned on  $\zeta(0)=\alpha(0)+j\beta(0)$, we have 
\begin{equation}\label{ch3.3.92}
{\mathrm{E}\left[\underline{\nr(\length)}\Big|\underline{\zeta(0)}=\zeta(0)\right]=0,
}
\end{equation}
\begin{equation}\label{ch3.3.94}
{
\mathrm{E}\left[\underline{\ni(\length)}\Big|\underline{\zeta(0)}=\zeta(0)\right]=\frac{1}{2}\epsilon^{2}\length,}
\end{equation}
\begin{equation}\label{ch3.3.95}
{\mathrm{E}\left[\underline{\nr(\length)}^{2}\bigg|\underline{\zeta(0)}=\zeta(0)\right]=\frac{1}{6}\epsilon^{2}\length\beta(0)+\frac{1}{24}\epsilon^{4}\length^{2},}
\end{equation}
\begin{equation}\label{ch3.3.93}
{\mathrm{E}\left[\underline{\ni(\length)}^{2}\bigg|\underline{\zeta(0)}=\zeta(0)\right]=\frac{1}{2}\epsilon^{2}\length\beta(0)+\frac{3}{8}\epsilon^{4}\length^{2}}.
\end{equation}
Please note again that the statistics as shown in \eqref{ch3.3.92}--\eqref{ch3.3.93} only hold to the first order in the noise level.

Applying the perturbation theory \eqref{ch3.3.6} and \eqref{ch3.3.2s}, we can study the dependency of the noises for the real and the imaginary parts of the eigenvalue. Again, this only holds to the first order.

\begin{thm}\label{ch3.t3.5}
{If the complex white Gaussian noise $\underline{G(t,z)}$ in \eqref{ch3.1.1} is circularly symmetric with zero mean, the dependency of the noise variables $\underline{\nr(z)}$ and $\underline{\ni(z)}$ (for any $z\in [0,\length]$) are summarised as follows:
\begin{enumerate}
{
\item Conditioned on an input eigenvalue $\zeta(0)$, we have
\begin{equation}
{\mathrm{E}\left[\underline{\nr(z)}\underline{\ni(z)}\Big|\underline{\zeta(0)}=\zeta(0)\right]=0}
\end{equation}
for any $z\in [0,\length]$;

\item $\underline{\nr(z)}$ and $\underline{\ni(z)}$ are uncorrelated for any $z\in[0,\length]$.
}
\end{enumerate}
}
\end{thm}

We note that if the fibre length is assumed to be short enough such that the distribution of the noise in an eigenvalue is approximately Gaussian, the random variables $\nr(z)$ and $\ni(z)$ are conditional independent for every $z\in[0,\length]$, because they are conditional jointly Gaussian distributed. Removing the Gaussian approximation, although no longer independent, they are still uncorrelated.


\begin{proof}[Proof of Theorem \ref{ch3.t3.5}]
{
For the simplicity of the notations, we denote 
\begin{equation}\label{appen.b.4.1}
W(t,z)\triangleq G(t,z)e^{-j\varphi(t,z)},
\end{equation}
and
\begin{equation}\label{appen.b.4.2}
s(t,z)\triangleq \beta(z)\mathrm{sech}\left[2\beta(z)\left(t-T_{0}(z)\right)\right],
\end{equation}
and
\begin{equation}
r(t,z)\triangleq \beta(z)\mathrm{sech}\left[2\beta(z)\left(t-T_{0}(z)\right)\right]\mathrm{tanh}\left[2\beta(z)\left(t-T_{0}(z)\right)\right].\label{appen.b.4.3}
\end{equation}
For any $\length>0$, we have
\begin{equation}
\underline{\ni(\length)}=\frac{1}{2}\epsilon^{2}\length+\epsilon\int_{0}^{\length}\int_{-\infty}^{\infty}\re\left[\underline{W(t,z)}\right]\underline{s(t,z)}\d t\d z,\label{ch3.a.4}
\end{equation}
and
\begin{equation}
\underline{\nr(\length)}=-\epsilon\int_{0}^{\length}\int_{-\infty}^{\infty}\im\left[\underline{W(t,z)}\right]\underline{r(t,z)}\d t\d z.\label{ch3.a.5}
\end{equation}
Denote
$$\textbf{Q}_{1,2}\triangleq\left(\varphi(t_{k},z_{k}), \beta(z_{k}),T_{0}(z_{k}):\ k=1,2\right),$$
and the conditional joint CDF  of $\underline{\textbf{Q}_{1,2}}\big|\underline{\zeta(0)}$ by $H_{6}$. Denote the support of the distribution $H_{6}$ by $S_{H_{6}}$. For any $z\in[0,\length]$, we have 
\begin{eqnarray}
&&\mathrm{E}\left[\underline{\nr(\length)}\underline{\nu_{I}(\length)}\Big|\underline{\zeta(0)}=\zeta(0)\right]\nonumber\\
 &=& \mathrm{E}\left\{-\int_{0}^{\length}\int_{0}^{\length}\iint_{\mathbb{R}^{2}}\im\Big[\underline{W(t_{1},z_{1})}\Big]\re\left[\underline{W(t_{2},z_{2})}\right]\underline{r(t_{1},z_{1})}\underline{s(t_{2},z_{2})}\d t_{1}\d t_{2}\d z_{1}\d z_{2}\Bigg|\underline{\zeta(0)}=\zeta(0)\right\}\nonumber\\
 &=& -\int_{0}^{\length}\int_{0}^{\length}\iint_{\mathbb{R}^{2}}\mathrm{E}\left\{\im\Big[\underline{W(t_{1},z_{1})}\Big]\re\left[\underline{W(t_{2},z_{2})}\right]\underline{r(t_{1},z_{1})}\underline{s(t_{2},z_{2})}\Big|\underline{\zeta(0)}=\zeta(0)\right\}\d t_{1}\d t_{2}\d z_{1}\d z_{2}\label{ch3.3.119}\\
 &=& -\int_{0}^{\length}\int_{0}^{\length}\iint_{\mathbb{R}^{2}}\idotsint_{S_{H_{6}}}\mathrm{E}\left\{\im\Big[\underline{W(t_{1},z_{1})}\Big]\re\left[\underline{W(t_{2},z_{2})}\right]\underline{r(t_{1},z_{1})}\underline{s(t_{2},z_{2})}\Big|\underline{\zeta(0)}=\zeta(0), \underline{\textbf{Q}_{1,2}}=\textbf{Q}_{1,2}\right\}\nonumber\\
 &&\d H_{6}\d t_{1}\d t_{2}\d z_{1}\d z_{2}\nonumber\\
&=& -\int_{0}^{\length}\int_{0}^{\length}\iint_{\mathbb{R}^{2}}\idotsint_{S_{H_{6}}}\mathrm{E}\left\{\im\Big[\underline{W(t_{1},z_{1})}\Big]\re\left[\underline{W(t_{2},z_{2})}\right]r(t_{1},z_{1})s(t_{2},z_{2})\Big|\underline{\zeta(0)}=\zeta(0), \underline{\textbf{Q}_{1,2}}=\textbf{Q}_{1,2}\right\}\nonumber\\
 &&\d H_{6}\d t_{1}\d t_{2}\d z_{1}\d z_{2}\nonumber\\
 &=& -\iint_{\{z_{1}\leq z_{2}\}}\iint_{\mathbb{R}^{2}}\idotsint_{S_{H_{6}}}\mathrm{E}\left\{\im\Big[\underline{G(t_{1},z_{1})}e^{-j\varphi(t_{1},z_{1})}\Big]r(t_{1},z_{1})s(t_{2},z_{2})\Big|\underline{\zeta(0)}=\zeta(0), \underline{\textbf{Q}_{1,2}}=\textbf{Q}_{1,2}\right\}\nonumber\\
&&\mathrm{E}\left\{\re\left[\underline{G(t_{2},z_{2})}e^{-j\varphi(t_{2},z_{2})}\right]\right\}\d H_{6}\d t_{1}\d t_{2}\d \sigma_{1}\nonumber\\
&&-\iint_{\{z_{1}> z_{2}\}}\iint_{\mathbb{R}^{2}}\idotsint_{S_{H_{6}}}\mathrm{E}\left\{\re\Big[\underline{G(t_{2},z_{2})}e^{-j\varphi(t_{2},z_{2})}\Big]r(t_{1},z_{1})s(t_{2},z_{2})\Big|\underline{\zeta(0)}=\zeta(0), \underline{\textbf{Q}_{1,2}}=\textbf{Q}_{1,2}\right\}\nonumber\\
&&\mathrm{E}\left\{\im\left[\underline{G(t_{1},z_{1})}e^{-j\varphi(t_{1},z_{1})}\right]\right\}\d H_{6}\d t_{1}\d t_{2}\d \sigma_{2}\label{ch3.a.6}\\
&=& 0+0=0.\label{ch3.a.57}
\end{eqnarray}
In equation \eqref{ch3.3.119}, we change the order of the \ito\ integral and the conditional expectation. This could be proved by realising that the order of the conditional expectation and the mean-square limit (of the partial sum of the \ito\ integral by definition) could be changed. Equation \eqref{ch3.a.6} holds because of Lemma \ref{ch3.l3.16}, and the definition of the \ito\ stochastic integral. Equation \eqref{ch3.a.57} stands because of Lemma \ref{ch3.l3.16}.

Furthermore, applying the Law of total expectation to \eqref{ch3.a.57}, we have
$$\mathrm{E}\left[\underline{\nr(\length)}\underline{\nu_{I}(\length)}\right]=0.$$
So we have
\begin{equation}
\mathrm{E}\left[\underline{\nr(\length)}\underline{\ni(\length)}\right]=\mathrm{E}\left[\underline{\nr(\length)}\underline{\nu_{I}(\length)}\right]+\frac{1}{2}\epsilon^{2}\length\cdot\mathrm{E}\left[\underline{\nr(\length)}\right]=0,\nonumber
\end{equation}
which means that $\underline{\nr(z)}$ and $\underline{\ni(z)}$ are uncorrelated for any $z\in[0,\length]$.

}
\end{proof}

\subsubsection{The Statistics of the Noise in the Magnitude of a Spectral Amplitude}

In this subsection, we study the non-Gaussian statistics of the noise in the magnitude of a spectral amplitude. Before this, we need the following lemma.


\begin{lemma}\label{ch3.l3.11}
{Denote 
\begin{equation}\label{ch3.3.124}
{\underline{\Gamma^{(R)}(\length)}\triangleq\int_{0}^{\length}\underline{\nr(z)}\d z,}
\end{equation}
\begin{equation}\label{ch3.3.125}
{\underline{\Gamma^{(I)}(\length)}\triangleq\int_{0}^{\length}\underline{\ni(z)}\d z,}
\end{equation}
and
\begin{equation}\label{ch3.3.126}
{\underline{\Gamma^{(RI)}(\length)}\triangleq\int_{0}^{\length}\underline{\nr(z)}\underline{\ni(z)}\d z.}
\end{equation}
Then conditioned on an input eigenvalue $\zeta(0)=\alpha(0)+j\beta(0)$, we have
\begin{equation}\label{ch3.3.54}
{\mathrm{E}\left[\underline{\Gamma^{(R)}(\length)}\underline{\Gamma^{(I)}(\length)}\Big|\underline{\zeta(0)}=\zeta(0)\right]=0,}
\end{equation}
and
\begin{equation}
\mathrm{E}\left[\underline{\Gamma^{(R)}(\length)}\underline{\Gamma^{(RI)}(\length)}\Big|\underline{\zeta(0)}=\zeta(0)\right]=\frac{5}{288}\epsilon^{4}\length^{4}\beta(0)+\frac{7}{2880}\epsilon^{6}\length^{5},\label{ch3.3.55}
\end{equation}
and
\begin{equation}\label{ch3.3.56}
{\mathrm{E}\left[\underline{\Gamma^{(I)}(\length)}\underline{\Gamma^{(RI)}(\length)}\Big|\underline{\zeta(0)}=\zeta(0)\right]=0.}
\end{equation}
}
\end{lemma}

\begin{proof}[Sketch of the Proof]
{By applying the technique used in Theorem \ref{ch3.t3.5}, this lemma can be proved. Please refer to Appendix \ref{AppendixB} for details.}
\end{proof}

The non-Gaussian noise statistics of the noise in the magnitude of the spectral amplitude of a soliton is stated and proved as follows.

\begin{thm}\label{ch3.t3.8}
{Denote the input eigenvalue random variable as $\zeta(0)=\alpha(0)+j\beta(0)$. The mean and variance of the noise in the spectral amplitude are as follows:
\begin{equation}\label{ch3.3.49}
{\mathrm{E}N(\length)=2\epsilon^{2}\length^{2}\mathrm{E}\alpha(0),}
\end{equation}
and
\begin{multline}
\mathrm{Var}\left[N(\length)\right]=\frac{32}{3}\epsilon^{2}\length^{3}\mathrm{E}\left[\alpha(0)^{2}\beta(0)\right]+\frac{32}{9}\epsilon^{2}\length^{3}\mathrm{E}\left[\beta(0)^{3}\right]+\frac{16}{3}\epsilon^{4}\length^{4}\mathrm{E}\left[\alpha(0)^{2}\right]+\frac{32}{9}\epsilon^{4}\length^{4}\mathrm{E}\left[\beta(0)^{2}\right]\\
+\frac{46}{45}\epsilon^{6}\length^{5}\mathrm{E}\beta(0)+\frac{23}{270}\epsilon^{8}\length^{6}-4\epsilon^{4}\length^{4}\left[\mathrm{E}\alpha(0)\right]^{2}.\label{ch3.3.50}
\end{multline}
}
\end{thm}



\begin{proof}[Proof of Theorem \ref{ch3.t3.8}]
{
1) We have
\begin{equation}
\mathrm{E}\underline{N(\length)}=8\mathrm{E}\left[\underline{\alpha(0)}\int_{0}^{\length}\underline{\ni(z)}\d z\right]+8\mathrm{E}\left[\underline{\beta(0)}\int_{0}^{\length}\underline{\nr(z)}\d z\right]+8\mathrm{E}\left[\int_{0}^{\length}\underline{\nr(z)}\underline{\ni(z)}\d z\right],\label{ch3.a.81}
\end{equation}
where we calculate
\begin{eqnarray}
\mathrm{E}\left[\underline{\alpha(0)}\int_{0}^{\length}\underline{\ni(z)}\d z\right]&=&\mathrm{E}\left\{\mathrm{E}\left[\underline{\alpha(0)}\int_{0}^{\length}\underline{\ni(z)}\d z\Bigg|\underline{\zeta(0)}\right]\right\}\label{ch3.a.82}\\
&=&\mathrm{E}\left\{\alpha(0)\int_{0}^{\length}\mathrm{E}\left[\underline{\ni(z)}\Big|\underline{\zeta(0)}=\zeta(0)\right]\d z\right\}\nonumber\\
&=&\frac{1}{4}\epsilon^{2}\length^{2}\mathrm{E}\alpha(0),\label{ch3.a.83}
\end{eqnarray}
where $\zeta(0)\triangleq\alpha(0)+j\beta(0)$. Equation \eqref{ch3.a.82} is obtained by applying the Law of total expectation, and \eqref{ch3.a.83} stands because of \eqref{ch3.3.94}. Similarly, we also have
\begin{equation}\label{ch3.a.84}
{\mathrm{E}\left[\underline{\beta(0)}\int_{0}^{\length}\underline{\nr(z)}\d z\right]=0.}
\end{equation} 
Plug \eqref{ch3.a.83}--\eqref{ch3.a.84} into \eqref{ch3.a.81}, we have
\begin{eqnarray}
\mathrm{E}\underline{N(\length)}&=&2\epsilon^{2}\length^{2}\mathrm{E}\alpha(0)+0+8\int_{0}^{\length}\mathrm{E}\left[\underline{\nr(z)}\underline{\ni(z)}\right]\d z\nonumber\\
&=&2\epsilon^{2}\length^{2}\mathrm{E}\alpha(0),\label{ch3.a.85}
\end{eqnarray}
where \eqref{ch3.a.85} is obtained by applying Theorem \ref{ch3.t3.5}.

2) We adopt the notations $\underline{\Gamma^{(R)}(\length)}$, $\underline{\Gamma^{(I)}(\length)}$, and $\underline{\Gamma^{(RI)}(\length)}$ as denoted in \eqref{ch3.3.124}--\eqref{ch3.3.126}, respectively. Then for a particular choice of the input eigenvalue $\zeta(0)=\alpha(0)+j\beta(0)$, we have 
\begin{eqnarray}
&&\mathrm{E}\left[\underline{\Gamma^{(I)}(\length)}^{2}\Big|\underline{\zeta(0)}=\zeta(0)\right]\nonumber\\
&=&\mathrm{E}\left[\int_{0}^{\length}\int_{0}^{\length}\underline{\ni(s)}\underline{\ni(t)}\d s\d t\Bigg|\underline{\zeta(0)}=\zeta(0)\right]\nonumber\\
&=&\int_{0}^{\length}\int_{0}^{\length}\mathrm{E}\left[\underline{\ni(s)}\underline{\ni(t)}\Big|\underline{\zeta(0)}=\zeta(0)\right]\d s\d t\nonumber\\
&=&2\iint_{\{s\leq t\}}\mathrm{E}\left[\underline{\ni(s)}^{2}\Big|\underline{\zeta(0)}=\zeta(0)\right]+\mathrm{E}\Big\{\underline{\ni(s)}\left[\underline{\ni(t)}-\underline{\ni(s)}\right]\Big|\underline{\zeta(0)}=\zeta(0)\Big\}\d\sigma\nonumber\\
&=&2\iint_{\{s\leq t\}}\mathrm{E}\left[\underline{\ni(s)}^{2}\Big|\underline{\zeta(0)}=\zeta(0)\right]+\frac{1}{4}\epsilon^{4}s(t-s)\d\sigma\label{ch3.a.89}\\
&=&2\iint_{\{s\leq t\}}\frac{1}{2}\epsilon^{2}s\beta(0)+\frac{3}{8}\epsilon^{4}s^{2}+\frac{1}{4}\epsilon^{4}s(t-s)\d\sigma\label{ch3.a.90}\\
&=&\frac{1}{12}\epsilon^{4}\length^{4}+\frac{1}{6}\epsilon^{2}\length^{3}\beta(0),\label{ch3.a.91}
\end{eqnarray}
where \eqref{ch3.a.89} is obtained because of 
\begin{eqnarray}
&&\mathrm{E}\left\{\underline{\ni(s)}\left[\underline{\ni(t)}-\underline{\ni(s)}\right]\Big|\underline{\zeta(0)}=\zeta(0)\right\}\nonumber\\
&=&\mathrm{E}\left\{\left[\underline{\nu_{I}(s)}+\frac{1}{2}\epsilon^{2}s\right]\left[\frac{1}{2}\epsilon^{2}(t-s)+\underline{\nu_{I}(t)}-\underline{\nu_{I}(s)}\right]\Big|\underline{\zeta(0)}=\zeta(0)\right\}\nonumber\\
&=&\frac{1}{4}\epsilon^{4}s(t-s),\label{ch3.3.110}
\end{eqnarray}
which can be proved using the same method of obtaining \eqref{ch3.3.62}. Equation \eqref{ch3.a.90} is obtained using \eqref{ch3.3.93}. Similarly. we can calculate
\begin{eqnarray}
\mathrm{E}\left[\underline{\Gamma^{(R)}(\length)}^{2}\Big|\underline{\zeta(0)}=\zeta(0)\right]&=&2\iint_{\{s\leq t\}}\mathrm{E}\left[\underline{\nr(s)}^{2}\Big|\underline{\zeta(0)}=\zeta(0)\right]\d\sigma\nonumber\\
&=&2\iint_{\{s\leq t\}}\left[\frac{1}{6}\epsilon^{2}s\beta(0)+\frac{1}{24}\epsilon^{4}s^{2}\right]\d\sigma\nonumber\\
&=&\frac{1}{18}\epsilon^{2}\length^{3}\beta(0)+\frac{1}{144}\epsilon^{4}\length^{4},\label{ch3.a.92}
\end{eqnarray}

We calculate $\mathrm{Var}\left[\underline{N(\length)}\right]$ 
\begin{eqnarray}
\mathrm{Var}\left[\underline{N(\length)}\right]&=&\mathrm{E}\left[\underline{N(\length)}^{2}\right]-\left[\mathrm{E}\underline{N(\length)}\right]^{2}\nonumber\\
&=& 64\mathrm{E}\left[\underline{\alpha(0)}\underline{\Gamma^{(I)}(\length)}+\underline{\beta(0)}\underline{\Gamma^{(R)}(\length)}+\underline{\Gamma^{(RI)}(\length)}\right]^{2}-\left[\mathrm{E}\underline{N(\length)}\right]^{2}\nonumber\\
&=& 64\bigg\{\mathrm{E}\left[\underline{\alpha(0)}\underline{\Gamma^{(I)}(\length)}\right]^{2}+\mathrm{E}\left[\underline{\beta(0)}\underline{\Gamma^{(R)}(\length)}\right]^{2}+\mathrm{E}\left[\underline{\Gamma^{(RI)}(\length)}\right]^{2}+2\mathrm{E}\left[\underline{\alpha(0)}\underline{\beta(0)}\underline{\Gamma^{(R)}(\length)}\underline{\Gamma^{(I)}(\length)}\right]+\nonumber\\
&&2\mathrm{E}\left[\underline{\alpha(0)}\underline{\Gamma^{(I)}(\length)}\underline{\Gamma^{(RI)}(\length)}\right]+2\mathrm{E}\left[\underline{\beta(0)}\underline{\Gamma^{(R)}(\length)}\underline{\Gamma^{(RI)}(\length)}\right]\bigg\}-\left[\mathrm{E}\underline{N(\length)}\right]^{2}.\label{ch3.a.86}
\end{eqnarray}
For a particular choice of eigenvalue $\zeta(0)=\alpha(0)+j\beta(0)$, consider
\begin{eqnarray}
\mathrm{E}\left[\underline{\alpha(0)}\underline{\Gamma^{(I)}(\length)}\right]^{2}&=&\mathrm{E}\left\{\mathrm{E}\left[\left(\underline{\alpha(0)}\underline{\Gamma^{(I)}(\length)}\right)^{2}\Big|\underline{\zeta(0)}\right]\right\}\label{ch3.a.87}\\
&=&\int_{S_{P_{\zeta(0)}}}\alpha(0)^{2}\mathrm{E}\left[\underline{\Gamma^{(I)}(\length)}^{2}\Big|\underline{\zeta(0)}=\zeta(0)\right]\d P_{\zeta(0)}\nonumber\\
&=&\int_{S_{P_{\zeta(0)}}}\alpha(0)^{2}\cdot\left(\frac{1}{6}\epsilon^{2}\length^{3}\beta(0)+\frac{1}{12}\epsilon^{4}\length^{4}\right)\d P_{\zeta(0)}\label{ch3.a.88}\\
&=&\frac{1}{6}\epsilon^{2}\length^{3}\mathrm{E}\left[\underline{\alpha(0)}^{2}\underline{\beta(0)}\right]+\frac{1}{12}\epsilon^{4}\length^{4}\mathrm{E}\left[\underline{\alpha(0)}^{2}\right],\label{ch3.a.93}
\end{eqnarray}
where we denote $P_{\zeta(0)}$ the distribution of $\underline{\zeta(0)}$, and $S_{P_{\zeta(0)}}$ its support. Equation \eqref{ch3.a.87} is obtained according to the Law of total expectation, and \eqref{ch3.a.88} is obtained using \eqref{ch3.a.91}. Similarly, we have
\begin{eqnarray}
\mathrm{E}\left[\underline{\beta(0)}\underline{\Gamma^{(R)}(\length)}\right]^{2}&=&\int_{S_{P_{\zeta(0)}}}\beta(0)^{2}\mathrm{E}\left[\underline{\Gamma^{(R)}(\length)}^{2}\Big|\underline{\zeta(0)}=\zeta(0)\right]\d P_{\zeta(0)}\nonumber\\
&=&\int_{S_{P_{\zeta(0)}}}\beta(0)^{2}\cdot\left(\frac{1}{18}\epsilon^{2}\length^{3}\beta(0)+\frac{1}{144}\epsilon^{4}\length^{4}\right)\d P_{\zeta(0)}\label{ch3.a.94}\\
&=&\frac{1}{18}\epsilon^{2}\length^{3}\mathrm{E}\left[\underline{\beta(0)}^{3}\right]+\frac{1}{144}\epsilon^{4}\length^{4}\mathrm{E}\left[\underline{\beta(0)}^{2}\right],\label{ch3.a.95}
\end{eqnarray}
where equation \eqref{ch3.a.94} is obtained using \eqref{ch3.a.92}. Furthermore, according to \eqref{ch3.3.109}, we have
\begin{eqnarray}
\mathrm{E}\left[\underline{\Gamma^{(RI)}(\length)}\right]^{2}&=&\mathrm{E}\left\{\int_{0}^{\length}\int_{0}^{\length}\underline{\nr(s)}\underline{\ni(s)}\underline{\nr(t)}\underline{\ni(t)}\d s\d t\right\}\nonumber\\
&=&\int_{0}^{\length}\int_{0}^{\length}\mathrm{E}\left[\underline{\nr(s)}\underline{\ni(s)}\underline{\nr(t)}\underline{\ni(t)}\right]\d s\d t\nonumber\\
&=&\frac{1}{72}\epsilon^{4}\length^{4}\mathrm{E}\left[\underline{\beta(0)}^{2}\right]+\frac{1}{90}\epsilon^{6}\length^{5}\mathrm{E}\underline{\beta(0)}+\frac{23}{17280}\epsilon^{8}\length^{6},\label{ch3.a.97}
\end{eqnarray}
According to Lemma \ref{ch3.l3.11}, using the Law of total expectation, we have
\begin{equation}
\mathrm{E}\left[\underline{\alpha(0)}\underline{\beta(0)}\underline{\Gamma^{(R)}(\length)}\underline{\Gamma^{(I)}(\length)}\right]=\mathrm{E}\left[\underline{\alpha(0)}\underline{\Gamma^{(I)}(\length)}\underline{\Gamma^{(RI)}(\length)}\right]=0,\label{ch3.a.98}
\end{equation}
and
\begin{equation}
\mathrm{E}\left[\underline{\beta(0)}\underline{\Gamma^{(R)}(\length)}\underline{\Gamma^{(RI)}(\length)}\right]=\frac{5}{288}\epsilon^{4}\length^{4}\mathrm{E}\left[\underline{\beta(0)}^{2}\right]+\frac{7}{2880}\epsilon^{6}\length^{5}\mathrm{E}\underline{\beta(0)}.\label{ch3.3.111}
\end{equation}

Plug \eqref{ch3.a.85}, \eqref{ch3.a.93} and \eqref{ch3.a.95}--\eqref{ch3.3.111} into \eqref{ch3.a.86}, we have
\begin{multline}
\mathrm{Var}\left[\underline{N(\length)}\right]=\frac{32}{3}\epsilon^{2}\length^{3}\mathrm{E}\left[\underline{\alpha(0)}^{2}\underline{\beta(0)}\right]+\frac{32}{9}\epsilon^{2}\length^{3}\mathrm{E}\left[\underline{\beta(0)}^{3}\right]+\frac{16}{3}\epsilon^{4}\length^{4}\mathrm{E}\left[\underline{\alpha(0)}^{2}\right]+\frac{32}{9}\epsilon^{4}\length^{4}\mathrm{E}\left[\underline{\beta(0)}^{2}\right]\\
+\frac{46}{45}\epsilon^{6}\length^{5}\mathrm{E}\underline{\beta(0)}+\frac{23}{270}\epsilon^{8}\length^{6}-4\epsilon^{4}\length^{4}\left[\mathrm{E}\underline{\alpha(0)}\right]^{2}.\nonumber
\end{multline}
}
\end{proof}

\subsection{Comparison with the Gordon-Haus Effect}\label{Section3.3.3}

In \cite{1986GH}, the Gordon-Haus effect was studied, which characterises the arrival time jitter of a soliton, i.e. the fluctuation of the centre of  a soliton. The centre of a soliton, denoted by $T_{0}(z)$, can be characterised by its eigenvalue and the magnitude of the spectral amplitude
\begin{equation}\label{ch3.x3.2}
T_{0}(z)=\frac{1}{2\beta(z)}\ln\frac{|Q^{(d)}(\zeta(z),z)|}{2\beta(z)}.
\end{equation}
By making an approximation that $\beta(z)\approx\beta(0)$, it could be observed that the variance of the time jittering and that of the fluctuation in the magnitude of the spectral amplitude are different only by a scalar for a given $\beta(0)$. As is suggested by the Gordon-Haus effect that the variance of the time jittering is $O(\length^{3})$ for long haul transmission, where $\length$ is the propagation distance, the variance of the noise in the magnitude of the spectral amplitude should be of order $O(\length^{3})$ as well for long distance communication.

As shown in Theorem \ref{ch3.t3.8}, our analysis provides a more accurate variance as is noticed that higher order terms, $O(\length^{6})$, are also included at a first glance. The main reasons therein are that unlike the assumptions in the Gordon-Haus effect, our channel model in Theorem \ref{ch3.t3.9} and Theorem \ref{ch3.t3.7} does not assume a Gaussian approximation, and that the perturbation of the imaginary part of the eigenvalue is also taken into account. However, it is shown later that the third order terms $O(\epsilon^{2}\length^{3})$ are indeed the dominant ones in real application scenario.

According to the typical fibre parameters summarised in Table \ref{ch1.table1.1} and the normalisation \eqref{ch2.x1.1} used, we have $\epsilon^{2}\approx 1.339\times 10^{-9}$, and $\length$ is of order $O(10^{3})$ to $O(10^{4})$ in normalised units as parameters of interest, where $\length$ corresponds to $10^{3}$ to $10^{4}$ $\mathrm{km}$ based on the normalisation. So $\epsilon^{2}\length$ is of order $10^{-5}$ to $10^{-6}$ in the normalised unit, which is quite small. On the other hand, we observe from \eqref{ch3.3.50} that the higher order terms in the variance are all $O((\epsilon^{2}\length^{3})\cdot(\epsilon^{2}\length)^{n})$, which are negligible compared to those of order $O(\epsilon^{2}\length^{3})$, $n=1,2,3$. As a result, we conclude that the variance of the magnitude of the spectral amplitude is of an order $O(\epsilon^{2}\length^{3})$, which matches the Gordon-Haus effect at high input power regime.

\subsection{Phase Channel}

Similar to the magnitude of a spectral amplitude, the noise model for the phase of a spectral amplitude, in terms of the statistics of the perturbation of an eigenvalue, is also obtained. Let $\arg Q^{(d)}(\zeta(0),\length)$ be the scale input of the phase of a spectral amplitude.

We take the imaginary part of the equation \eqref{ch3.3.39}, after the simple but tedious verification of the mean-square integrability similar to what we do in the proof of Theorem \ref{ch3.t3.7}, we could obtain the noisy input-and-output relationship for the phase of a spectral amplitude as follows

\begin{multline}
\arg Q^{(d)}(\zeta(\length),\length)=\arg Q^{(d)}(\zeta(0),\length)-8\alpha(0)\int_{0}^{\length}\nr(z)\d z+8\beta(0)\int_{0}^{\length}\ni(z)\d z\\
-4\int_{0}^{\length}\left[\nr(z)^{2}-\ni(z)^{2}\right]\d z+2s\pi,\label{ch3.4.100}
\end{multline}
where $s$ is a random variable taking integer values, or equivalently,
\begin{multline}
\arg Q^{(d)}(\zeta(\length),\length)\equiv\arg Q^{(d)}(\zeta(0),\length)-8\alpha(0)\int_{0}^{\length}\nr(z)\d z+8\beta(0)\int_{0}^{\length}\ni(z)\d z\\
-4\int_{0}^{\length}\left[\nr(z)^{2}-\ni(z)^{2}\right]\d z\qquad(\mathrm{mod}\ 2\pi).\label{ch3.4.101}
\end{multline}
 
However, the non-Gaussian noise statistics of the phase of a spectral amplitude are not available. The reason is that the statistics of the wrapped distribution need the distribution of the corresponding unwrapped one, which is not available because the distributions of the stochastic processes $\nr(z)$ and $\ni(z)$ are unknown, $z\in[0,\length]$.

\section{Discussions}\label{Discussions}

So far, we have described the noise modelling of discrete spectral amplitudes of $N$-solitons. The main idea of the modelling methodology is to regard the optical fibre channel as many short segments concatenated together, such that a simplifying ``two-phase'' channel modelling approach could be applied to each segment. We are interested in the type of noise in the second phase, which is the accumulation of the eigenvalue perturbation during the noiseless spatial evolution of the spectral amplitude.



In this section, we show that the noise in the second phase is significant, or even the dominant part in some scenarios of interest for soliton inputs, i.e. long haul transmission in a high input power regime. This further motivates us to conjecture that our model captures a significant noise for $N$-solitons as well in these scenarios. The approach is that we start with the classic soliton perturbation theory, and derive another noise model of the spectral amplitude of a soliton \eqref{ch3.x3.4}. Then we compare it with the special case ($N=1$) of the model we obtained in Theorem \ref{ch3.t3.9}, and show that the noise in our model \eqref{ch3.3.46} is significant for the scenarios of interest. We note that only a special case of the model, i.e. soliton inputs, can be derived from perturbation theory, however, our proposed methodology leads to the model for $N$-solitons where $N\geq 1$.

We also notice a recent related work \cite{2016JLT/HK}. An algorithm for a more accurate calculation of discrete spectral amplitudes was proposed in \cite{2016JLT/HK}, and the models studied in this paper (reported earlier as a preliminary work in \cite{Zhan1506:ASpectral}) were used as examples of implementing this algorithm in the scenario of long distance communication. For general $N$-soliton inputs, it was shown \cite{2016JLT/HK} that the noises in the magnitude and phase of a discrete spectral amplitude are respectively strongly correlated with the noises in its corresponding eigenvalue. This gives us insights that our models \eqref{ch3.3.46n} and \eqref{ch3.4.102} characterise the main factors of the noise in the discrete spectral amplitudes of an $N$-soliton ($N\geq 2$). Furthermore, for soliton inputs, it was shown \cite{2016JLT/HK} that the noise statistics, mean and variance, meet our theoretical prediction obtained in Theorem \ref{ch3.t3.8} for typical fibre parameters summarised in Table \ref{ch1.table1.1}, which is much in line with our discussion in Section \ref{Section3.3.3} that the noise variance is of order $O(\length^{3})$ for long haul transmission.

\subsection{From Perturbation Theory}

When the input signals are solitons, the perturbation of the centre of a soliton, which is closely related to its spectral amplitude \eqref{ch3.x3.2}, is available \eqref{ch3.3.32}. 

\subsubsection{Channel Model}

Using \eqref{ch3.x3.2}, we have
\begin{eqnarray}
\ln|Q^{(d)}(\zeta(\length),\length)|-\ln|Q^{(d)}(\zeta(0),\length)|&=& 2\beta(\length)T_{0}(\length)+\ln\left[2\beta(\length)\right]-2\beta(0)T_{0}(0)-\ln\left[2\beta(0)\right]\nonumber\\
&=& 2\beta(\length)T_{0}(\length)-2\left[\beta(\length)-\ni(\length)\right]T_{0}(0)+\ln\left[\frac{\beta(\length)}{\beta(0)}\right]\nonumber\\
&=& 2\beta(\length)\left[T_{0}(\length)-T_{0}(0)\right]+2\ni(\length)T_{0}(0)\nonumber\\
&&+\ln\left[\frac{\beta(\length)}{\beta(0)}\right].\label{ch3.x3.3}
\end{eqnarray}

Integrating the perturbation theory \eqref{ch3.3.32} from $[0,\length]$, and insert it to \eqref{ch3.x3.3}, we have
\begin{multline}
\ln|Q^{(d)}(\zeta(\length),\length)|-\ln|Q^{(d)}(\zeta(0),\length)|=8\length\alpha(0)\ni(\length)+8\beta(0)\int_{0}^{\length}\nr(z)\d z\\
+8\ni(\length)\int_{0}^{\length}\nr(z)\d z+2\ni(\length)\int_{0}^{\length}\Delta(z)\d z\\
+2\beta(0)\int_{0}^{\length}\Delta(z)\d z+2\ni(\length)T_{0}(0)+\ln\left[\frac{\beta(\length)}{\beta(0)}\right],\label{ch3.x3.4}
\end{multline}
where $Q^{(d)}(\zeta(0),\length)$ is defined in \eqref{ch3.x3.1}, and
\begin{equation}
\Delta(z)\triangleq \epsilon\int_{-\infty}^{\infty}\re\left[G(t,z)e^{-j\varphi(t,z)}\right](t-T_{0}(z))\mathrm{sech}[2\beta(z)(t-T_{0}(z))]\d t.\label{ch3.x3.5}
\end{equation}
So \eqref{ch3.x3.4} describes the noise in the magnitude of the spectral amplitude of a soliton suggested by the perturbation theory. Denote
\begin{equation}\label{ch3.x3.6} 
 \begin{split} 
     N_{11}&\triangleq 2\beta(0)\int_{0}^{\length}\Delta(z)\d z\\ 
     N_{12}&\triangleq 2\ni(\length)T_{0}(0) \\
     N_{13}&\triangleq\ln\left[\frac{\beta(\length)}{\beta(0)}\right],
 \end{split} 
\end{equation} 
and
\begin{equation}\label{ch3.x3.7}
N_{2}\triangleq 2\ni(\length)\int_{0}^{\length}\Delta(z)\d z,
\end{equation}
and
\begin{equation}\label{ch3.x3.8}
N_{31}\triangleq 8\length\alpha(0)\ni(\length),\qquad N_{32}\triangleq 8\beta(0)\int_{0}^{\length}\nr(z)\d z
\end{equation}
and
\begin{equation}\label{ch3.x3.9}
N_{4}\triangleq 8\ni(\length)\int_{0}^{\length}\nr(z)\d z.
\end{equation}
We further denote
\begin{equation}\label{ch3.x3.10}
N_{1}\triangleq N_{11}+N_{12}+N_{13},
\end{equation}
and
\begin{equation}\label{ch3.x3.11}
N_{3}\triangleq N_{31}+N_{32},
\end{equation}
then equation \eqref{ch3.x3.4} could be denoted as
\begin{equation}\label{ch3.x3.12}
\ln|Q^{(d)}(\zeta(\length),\length)|-\ln|Q^{(d)}(\zeta(0),\length)|=N_{1}+N_{2}+N_{3}+N_{4}.
\end{equation}

\subsubsection{Noise Statistics}

The noise statistics are summarised as follows. 
\begin{equation}\label{ch3.x3.13} 
    \mathrm{E}N_{1}\approx 0\qquad 
    \mathrm{Var}N_{1}\approx 0.912\epsilon^{2}\length\mathrm{E}\left[\frac{1}{\beta(0)}\right]+2\epsilon^{2}\length\mathrm{E}\left[\beta(0)T_{0}(0)^{2}\right]+2\epsilon^{2}\length\mathrm{E}\left[T_{0}(0)\right],
\end{equation} 
\begin{equation}\label{ch3.x3.14}
\mathrm{E}N_{2}= 0,\qquad \mathrm{Var}N_{2}\approx 0.206\epsilon^{4}\length^{2}\mathrm{E}\left[\frac{1}{\beta(0)^{2}}\right],
\end{equation}
\begin{equation}\label{ch3.x3.15} 
    \mathrm{E}N_{3}\approx 0\qquad 
    \mathrm{Var}N_{3}\approx 32\epsilon^{2}\length^{3}\mathrm{E}\left[\alpha(0)^{2}\beta(0)\right]+\frac{32}{9}\epsilon^{2}\length^{3}\mathrm{E}\left[\beta(0)^{3}\right],
\end{equation} 
\begin{equation}\label{ch3.x3.16} 
    \mathrm{E}N_{4}= 0\qquad \mathrm{Var}N_{4}= \frac{16}{9}\epsilon^{4}\length^{4}\mathrm{E}\left[\beta(0)^{2}\right]+\frac{68}{45}\epsilon^{6}\length^{5}\mathrm{E}\left[\beta(0)\right]+\frac{16}{135}\epsilon^{8}\length^{6}. 
\end{equation} 
The covariances are
\begin{equation}\label{ch3.x3.17}
\mathrm{Cov}(N_{1},N_{3})\approx 8\epsilon^{2}\length^{2}\mathrm{E}\left[\alpha(0)\beta(0)T_{0}(0)\right]+4\epsilon^{2}\length^{2}\mathrm{E}\alpha(0),
\end{equation}
and
\begin{equation}\label{ch3.x3.18}
\mathrm{Cov}(N_{m},N_{n})\approx 0,\qquad 1\leq m<n\leq 4,\ \textrm{and}\ (m,n)\neq (1,3).
\end{equation}

The proofs of \eqref{ch3.x3.13}--\eqref{ch3.x3.18} are  essentially similar to those of Lemma \ref{ch3.l3.9} and Theorem \ref{ch3.t3.8}, however, they are more complicated and difficult. We omit the tedious derivation, and only describe the approximations we make.

We note that the statistics are derived based on two typical approximations, which are often used in existing works, such as \cite{1996HW}. The first one is that we approximate the terms $\beta(z)$ in the soliton perturbation theory, i.e. the right hand sides of \eqref{ch3.3.6}--\eqref{ch3.3.2s}, as $\beta(0)$. The second one is that we omit the advection term $\frac{1}{2}\epsilon^{2}$ in the \ito\ stochastic perturbation theory \eqref{ch3.3.2s}. 
The resulting statistics are still non-Gaussian because other soliton parameters, such as $T_{0}(z)$ and $\varphi(t,z)$, are not approximated, making the noise $G(t,z)$ on the right hand sides of \eqref{ch3.3.6}--\eqref{ch3.3.2s} still multiplicative rather than additive.



Similar to the discussion in Section \ref{Section3.3.3}, the terms $N_{2}$ and $N_{4}$ are negligible compared to $N_{1}$ and $N_{3}$, respectively, in a high input power regime for long haul transmission using the typical fibre parameters in Table \ref{ch1.table1.1}. In addition, the term $N_{3}$ is more significant than $N_{1}$ in the same scenario. It will be shown intuitively in Example \ref{ch3.xe4.1} at the end of this section.

\subsection{Comparison with Our Model}

In this subsection, we compare the model \eqref{ch3.3.46} derived from our modelling methodology with \eqref{ch3.x3.4} to discuss that the noise captured in our model \eqref{ch3.3.46} is significant in some scenarios of interest. The discussion is for soliton inputs, however, we expect it could also provide insights of modelling the noise in discrete spectral amplitudes of more general classes of signals, such as $N$-solitons as well.

We notice that our model \eqref{ch3.3.46} characterises the noise in a spectral amplitude for the long haul communication scenario, and describes the Gordon-Haus effect after comparing with the $N_{3}$ term in \eqref{ch3.x3.4}. In addition, according to Theorem \ref{ch3.t3.8}, the dominant noise terms in our model is
\begin{equation}\label{ch3.x3.20}
N_{0}\triangleq 8\alpha(0)\int_{0}^{\length}\ni(z)\d z+8\beta(0)\int_{0}^{\length}\nr(z)\d z
\end{equation}
using the discussion in Section \ref{Section3.3.3}. We find out that the noise term $N_{0}$ is a good approximation of $N_{3}$ when the input $\alpha(0)$ is very close to $0$. This is in general a good choice in applications. Specifically, since $\alpha(0)$ determines the velocity of each pulse, having a large range for the input alphabet of $\alpha(0)$ induces significant relative motions in a soliton train, which is not desirable for long haul transmission because the order of solitons in a train may change, resulting in the difficulties in decoding. Having the input alphabet of $\alpha(0)$ concentrated is an easy way to avoid this issue. In addition, it might be desirable to have $\alpha(0)$ concentrated at $0$ because it reduces the noise variances of both $N_{0}$ and $N_{3}$ compared to other choices of $\alpha(0)$.

In summary, the non-Gaussian soliton spectral amplitude noise model \eqref{ch3.3.46} is a good approximation of \eqref{ch3.x3.4} obtained from perturbation theory when $\alpha(0)$ is close to $0$, which also characterises a significant noise for long haul transmission. Long distance communication is an interesting and important application scenario, corresponding to transoceanic data transmission. 
In the following, we give an example to intuitively show that $N_{3}$, and hence $N_{0}$, is significant in this application scenario of interest.


We take actual soliton communication system into consideration, and use typical fibre parameters in Table \ref{ch1.table1.1} for analysis. In soliton communication systems, soliton pulses in a soliton train should not be too close to each other because of the soliton interactions. It was shown \cite{b43} that the soliton pulses in a train for transmission must be separated by at least several pulse widths in order to mitigate the interaction effects. In addition, a limitation on the propagation distance is required to maintain the interaction at a sufficiently low level \cite{1983Gordon}. 
As a result, a soliton separation of five or more pulse widths is often allocated to mitigate the interaction effect. The pulse width is interpreted as full width at half maximum (FWHM) pulse width.

In order to encode the real part of the eigenvalue and the discrete spectral amplitude (or pulse centre), additional soliton separation is needed. In particular, 
we allocate two additional FWHM pulse widths separation 
 on top of the ones allocated to mitigate the soliton interaction.
As a result, the pulse separation in the following example is set as seven FWHM pulse widths. 

\begin{example}\label{ch3.xe4.1}
The parameters used is summarised in Table \ref{ch1.table1.1} and Table \ref{ch3.xtable4.1}, where the input power is defined as the signal energy over the (time) separation of the neighboring two solitons, which is seven FWHM pulse widths. The symbol rate is approximately $1.435$ ${\mathrm{Giga}\ \mathrm{symbol}}/{\mathrm{s}}$ for an input power of $0.8$ $\mathrm{mW}$. For mutually independent uniformly distributed inputs $\alpha(0)$, $\beta(0)$, and $T_{0}(0)$ with $\mathrm{E}\alpha(0)=0$, we have $\mathrm{Var}N_{3}\geq\frac{32}{9}\epsilon^{2}\length^{3}\mathrm{E}\left[\beta(0)^{3}\right]$, and $\mathrm{Var}N_{0}\geq\frac{32}{9}\epsilon^{2}\length^{3}\mathrm{E}\left[\beta(0)^{3}\right]$, and the ratio of the variance of $N_{3}$ and that of $N_{1}$ is
\begin{equation}\label{ch3.x3.19}
r\triangleq\frac{\mathrm{Var}N_{3}}{\mathrm{Var}N_{1}}\approx 103.45,
\end{equation}
which also gives the approximation of the ratio of the variances of $N_{0}$ and $N_{1}$.

\begin{table*}[!t]
\renewcommand{\arraystretch}{1.3}
\caption{Parameters for A Soliton Communication Channel}
\label{ch3.xtable4.1}
\centering
\begin{tabular}{|c|c|c|}
\hline
Symbols & Values & Explanations\\
\hline
$\length$ & $7000$ $\mathrm{km}$ & Propagation distance \\
$P_{0}$ & $0.8$ $\mathrm{mW}$ & Input power \\
$b$ & $0.028$  & Corresponding imaginary part of the eigenvalue \\
$T_{0.5}$ & $1.763$ (normalised unit)& FWHM pulse width for soliton $\mathrm{sech}(t)$\\
$\beta(0)$ & $\beta(0)\in[0.9b, 1.1b]$ & Input alphabet of $\beta(0)$\\
$T_{0}(0)$ & $T_{0}(0)\in\left[-\frac{T_{0.5}}{4b},\frac{T_{0.5}}{4b}\right]$ & Input alphabet of $T_{0}(0)$\\
\hline
\end{tabular} 
\end{table*}
\end{example}

In summary, for a propagation distance of $7000$ $\mathrm{km}$, the variance of $N_{0}$ (and also $N_{3}$) is roughly $100$ times bigger than that of $N_{1}$ at an input power $0.8$ $\mathrm{mW}$. The value of $r$ increases as the input power increases. For example, $r\approx 1010.27$ if we increase the input power to $2.5$ $\mathrm{mW}$ with the transmission distance unchanged. Example \ref{ch3.xe4.1} gives some intuitive insights that the noise term $N_{0}$ is significant, or even dominant in long houl transmission in high input power regime.

\section{Conclusions}\label{Conclusions}

In this paper, we propose analytic models for the noisy evolution of the spectral amplitudes for $N$-solitons. Deriving such a model is the very first and also an important step towards developing a communication system that conveys information via spectral amplitudes, which has the potential to greatly increase the communication rate in fibre-optic communications. 

We propose a noise modelling approach, in which the optical fibre is treated as many segments concatenated together, and the noise along each segment is modelled by two phases. In the first phase, both the eigenvalue and the spectral amplitude are directly perturbed by the white Gaussian noise \eqref{ch3.1.1}. In the second phase, there is no noise in the fibre segment, however, the noise in the eigenvalue occurred in the first phase accumulates additional noise in the spectral amplitude during the second phase suggested by its spatial evolution. 
We focus on the noise in the second phase, and derive an analytical noise model for the discrete spectral amplitudes of an $N$-soliton. We find out that the second phase noise is significant, or even dominant, in some scenarios where the propagation distance is long for soliton inputs using the perturbation theory. We give an example intuitively showing this observation. Motivated by this observation, we also focus on this noise for $N$-soliton inputs, and conjecture that it could also be significant for long haus transmission in a high input power regime.

\appendices

\section{Mathematical Tools}\label{AppendixB}

The lemmas proved in this appendix are all mathematical preliminaries for the proof of the main theorems of this paper. We list them here separately because they can be independent themselves as mathematical results. 



\begin{lemma}\label{ch3.l3.16}
{Let $G(t,z)$ be a zero mean circularly symmetric \cite{b29} complex white Gaussian noise process. Denote
\begin{equation}\label{ch3.3.97}
{W(t,z)\triangleq G(t,z)e^{-j\varphi(t,z)},\qquad\forall\ t\in\mathbb{R}, z\in[0,\length],}
\end{equation}
where $\varphi(t,z)$ is a real deterministic function. Then $W(t,z)$ is also a circularly symmetric white Gaussian noise process. Furthermore, for all positive integers $l$ and all choices of epochs $(t_{k},z_{k})$, $k=1,2,\ldots,l$, the set of random variables $\{\re[W(t_{k},z_{k})],\im[W(t_{k},z_{k})]:\ k=1,2,\ldots,l\}$ are mutually independent.
}
\end{lemma}

\begin{proof}[Proof of Lemma \ref{ch3.l3.16}]
{For any positive integer $l$ and all choice of epochs $(t_{k},z_{k})$, $k=1,2,\ldots,l$, we first prove that $(W(t_{k},z_{k}),\ k=1,2,\ldots,l)$ are jointly Gaussian. To show this, denote
\begin{equation}
\mathbf{y}\triangleq (\re[W(t_{1},z_{1})],\im[W(t_{1},z_{1})],\ldots,\re[W(t_{l},z_{l})],\im[W(t_{l},z_{l})])^{\mathrm{T}},\label{ch3.a.118}
\end{equation}
and
\begin{equation}
\mathbf{x}\triangleq (\re[G(t_{1},z_{1})],\im[G(t_{1},z_{1})],\ldots,\re[G(t_{l},z_{l})],\im[G(t_{l},z_{l})])^{\mathrm{T}}.\label{ch3.a.119}
\end{equation}
Then we have
\begin{equation}\label{ch3.a.120}
{\mathbf{y}=A\mathbf{x},
}
\end{equation}
where $\mathbf{0}$ is a $2\times 2$ zero matrix, and
\begin{equation}\label{ch3.a.127}
A=\left(
\begin{array}{cccc}
 A_{1} & \mathbf{0} & \ldots & \mathbf{0} \\
 \mathbf{0} & A_{2} & \ldots & \mathbf{0} \\
 \ldots & \ldots & \ldots & \ldots \\
 \mathbf{0} & \mathbf{0} & \ldots & A_{l}
\end{array}
\right)_{l\times l},
\end{equation}
and
\begin{equation}\label{ch3.a.121}
{A_{k}\triangleq\left(
\begin{array}{cc}
\cos\varphi(t_{k},z_{k}) & \sin\varphi(t_{k},z_{k})  \\
 -\sin\varphi(t_{k},z_{k}) & \cos\varphi(t_{k},z_{k}) 
\end{array}
\right)
}
\end{equation}
for $k=1,2,\ldots,l$. Since $G(t,z)$ is a zero mean complex white Gaussian process, $\mathbf{x}$ is a zero mean real Gaussian $2l$-random vector. So \eqref{ch3.a.120} gives us that $\mathbf{y}$ is a real Gaussian $2l$-random vector. So $\mathbf{v}\triangleq (W(t_{k},z_{k}),\ k=1,2,\ldots,l)^{\mathrm{T}}$ is a zero mean complex Gaussian random vector.

To show that $\mathbf{v}$ is circularly symmetric, denote $\mathbf{u}\triangleq (G(t_{k},z_{k}),\ k=1,2,\ldots,l)^{\mathrm{T}}$, and
\begin{equation}\label{ch3.a.129}
D=\left(
\begin{array}{cccc}
 e^{-j\varphi(t_{1},z_{1})} & 0 & \ldots & 0 \\
 0 &  e^{-j\varphi(t_{2},z_{2})} & \ldots & 0 \\
 \ldots & \ldots & \ldots & \ldots \\
 0 & 0 & \ldots & e^{-j\varphi(t_{l},z_{l}})
\end{array}
\right)_{l\times l},
\end{equation}
then we have  $\mathbf{v}=D\mathbf{u}$. Consider the pseudo-covariance matrix \cite{b29} of $\mathbf{v}$, which is
\begin{eqnarray}
\mathrm{E}\left[\mathbf{v}\mathbf{v}^{\textbf{T}}\right]&=&\mathrm{E}\left[D\mathbf{u}\mathbf{u}^{\textbf{T}}D^{\textbf{T}}\right]\nonumber\\
&=&D\mathrm{E}\left[\mathbf{u}\mathbf{u}^{\textbf{T}}\right]D^{\textbf{T}}\nonumber\\
&=& 0,\label{ch3.a.128}
\end{eqnarray}
where we have \eqref{ch3.a.128} because $G(t,z)$ is circularly symmetric which guarantees that $\mathrm{E}\left[\mathbf{u}\mathbf{u}^{\textbf{T}}\right]=0$. Consequently, $W(t,z)$ is a zero mean circularly symmetric complex Gaussian stochastic process.

To show that $W(t,z)$ is also white, consider
\begin{eqnarray}
\mathrm{E}\left[W(t_{1},z_{1})W^{*}(t_{2},z_{2})\right]&=&e^{-j[\varphi(t_{1},z_{1})-\varphi(t_{2},z_{2})]}\mathrm{E}\left[G(t_{1},z_{1})G^{*}(t_{2},z_{2})\right]\nonumber\\
&=&\delta(t_{1}-t_{2})\delta(z_{1}-z_{2}),\label{ch3.a.122}
\end{eqnarray}
where we use the superscript ``$*$'' to denote the conjugate of a complex number. Equation \eqref{ch3.a.122} holds because $G(t,z)$ is a white stochastic process.

Since the random vector $\{W(t_{k},z_{k}):\ k=1,2,\ldots,l\}$ is a zero mean circularly symmetric complex Gaussian random vector for any positive integers $l$ and any choices of epochs $(t_{k},z_{k})$, $k=1,2,\ldots,l$. By noticing its pseudo-covariance matrix \eqref{ch3.a.128} and equation \eqref{ch3.a.122}, it could be easily shown that any two components in $\mathbf{y}$ are pairwise independent. Hence, the components in $\textbf{y}$ are mutually independent because they are jointly Gaussian distributed.
}
\end{proof}

\begin{lemma}[\cite{b18}]\label{ch3.l3.12}
{Let $X_{1},X_{2},X_{3},X_{4}$ be jointly Gaussian random variables all with zero mean. Denote $\sigma_{st}\triangleq\mathrm{E}(X_{s}X_{t})$, where $s,t\in\{1,2,3,4\}$. Then we have
\begin{equation}\label{ch3.3.98}
{\mathrm{E}[X_{1}X_{2}X_{3}]=0,}
\end{equation}
and
\begin{equation}\label{ch3.3.76}
{\mathrm{E}[X_{1}X_{2}X_{3}X_{4}]=\sigma_{12}\sigma_{34}+\sigma_{13}\sigma_{24}+\sigma_{14}\sigma_{23}.}
\end{equation}
}
\end{lemma}


\begin{proof}[Proof of Lemma \ref{ch3.l3.12}]
{Equation \eqref{ch3.3.76} is actually an exercise in the classical digital communication textbook \cite{b18}. Equation \eqref{ch3.3.98} can be proved using the same method of proving \eqref{ch3.3.76}. Here we concisely highlight the idea of the proof.

Denote the jointly characteristic function of the $k$-dimensional zero mean Gaussian random vector $(X_{1},\ldots,X_{k})$ by $f(t_{1},\ldots,t_{k})$, where $k=3,4$. We have
\begin{equation}\label{ch3.a.123}
{f(\mathbf{t})=\exp\left(-\frac{1}{2}\mathbf{t}^{\mathrm{T}}\Lambda\mathbf{t}\right),}
\end{equation}
where $\mathbf{t}\triangleq(t_{1},\ldots,t_{k})^{\mathrm{T}}$, and $\Lambda\triangleq(\sigma_{mn})_{k\times k}$ is a $k\times k$ matrix, and $\sigma_{mn}\triangleq\mathrm{E}(X_{m}X_{n})$, for $m,n=1,2,\ldots,k$. In addition, we have
\begin{equation}\label{ch3.a.124}
{\mathrm{E}\left(\prod_{s=1}^{k}X_{s}\right)=j^{k}\frac{\partial^{k}f(\mathbf{t})}{\partial t_{1}\cdots\partial t_{k}}\bigg|_{t_{1}=\cdots=t_{k}=0}.}
\end{equation}
This lemma can be immediately proved by noticing the following two equations
\begin{equation}\label{ch3.a.125}
{\frac{\partial f(\mathbf{t})}{\partial t_{m}}=-\mathbf{\sigma}_{m}^{\mathrm{T}}\mathbf{t}\exp\left(-\frac{1}{2}\mathbf{t}^{\mathrm{T}}\Lambda\mathbf{t}\right),}
\end{equation}
and
\begin{equation}\label{ch3.a.126}
{\frac{\partial\mathbf{\sigma}_{m}^{\mathrm{T}}\mathbf{t} }{\partial t_{n}}=\sigma_{nm}=\sigma_{mn},}
\end{equation}
where $\mathbf{\sigma}_{m}=(\sigma_{i1},\ldots,\sigma_{ik})^{\mathrm{T}}$, and $m,n=1,2,\ldots,k$.
}
\end{proof}

\begin{lemma}\label{ch3.l3.10}
{Let $\{X_{n}\}$ and $\{Y_{n}\}$ be two sequences of random variables. Let $c$ be a real constant. If $\{X_{n}\}$ and $\{Y_{n}\}$ mean-square converge to random variables $X$ and $Y$ respectively, as $n\rightarrow\infty$, then we have
\begin{enumerate}
{
\item $X_{n}+Y_{n}$ mean-square converges to $X+Y$, as $n\rightarrow\infty$;
\item $cX_{n}$ mean-square converges to $cX$, as $n\rightarrow\infty$.

}
\end{enumerate}
}
\end{lemma}





\begin{proof}[Proof of Lemma \ref{ch3.l3.10}]
{Since $X_{n}$ mean-square converges to $X$, we have
\begin{equation}\label{ch3.a.70}
{\mathrm{E}|X_{n}-X|^{2}\rightarrow 0,\qquad n\rightarrow\infty.}
\end{equation}
1) We have
\begin{equation}
\mathrm{E}|(X_{n}+Y_{n})-(X+Y)|^{2}=\mathrm{E}(X_{n}-X)^{2}+\mathrm{E}(Y_{n}-Y)^{2}+2\mathrm{E}[(X_{n}-X)(Y_{n}-Y)].\label{ch3.a.71}
\end{equation}
According to the Cauchy's Inequality,
\begin{equation}\label{ch3.a.79}
{\mathrm{E}[(X_{n}-X)(Y_{n}-Y)]\leq\left[\mathrm{E}(X_{n}-X)^{2}\right]^{\frac{1}{2}}\left[\mathrm{E}(Y_{n}-Y)^{2}\right]^{\frac{1}{2}}.}
\end{equation}
Plug \eqref{ch3.a.79} into \eqref{ch3.a.71}, we have
\begin{multline}
0\leq\mathrm{E}|(X_{n}+Y_{n})-(X+Y)|^{2}\leq
\mathrm{E}(X_{n}-X)^{2}+\mathrm{E}(Y_{n}-Y)^{2}\\
+2\left[\mathrm{E}(X_{n}-X)^{2}\right]^{\frac{1}{2}}\left[\mathrm{E}(Y_{n}-Y)^{2}\right]^{\frac{1}{2}}\rightarrow 0\label{ch3.a.80}
\end{multline}
as $n$ goes to infinity, because $X_{n}$ and $Y_{n}$ mean-square converge to $X$ and $Y$, respectively. So $X_{n}+Y_{n}$ mean-square converges to $X+Y$.

2) Since $c\in\mathrm{R}$, we have
\begin{equation}
{0\leq\mathrm{E}|cX_{n}-cX|^{2}=c^{2}\mathrm{E}(X_{n}-X)^{2}\rightarrow 0,\qquad n\rightarrow\infty,}
\end{equation}
because $X_{n}$ mean-square converges to $X$. So $cX_{n}$ mean-square converges to $cX$.
}
\end{proof}

\section{Proof of Lemma \ref{ch3.l3.9}}\label{AppendixB2}

\begin{proof}
{
For notation simplicity, we adopt the notations \eqref{appen.b.4.1}--\eqref{appen.b.4.3}. Denote
\begin{equation}\label{appen.b.4.5}
\textbf{Q}_{k}\triangleq\left(\beta(z_{k}),T_{0}(z_{k}),\varphi(t_{k},z_{k})\right),\qquad k=1,2,3,4.
\end{equation}

1) Since $s\leq t$, we have 
\begin{eqnarray}
&&\mathrm{E}\left\{\nr(s)^{2}\nu_{I}(s)\left[\nu_{I}(t)-\nu_{I}(s)\right]\right\}\nonumber\\
&=&\epsilon^{4}\mathrm{E}\iiint_{[0,s)^{3}}\int_{[s,t)}\iiiint_{\mathbb{R}^{4}}\prod_{k=1}^{2}\left\{\im\left[W(t_{k},z_{k})\right]r(t_{k},z_{k})\right\}\prod_{k=3}^{4}\left\{\re\left[W(t_{k},z_{k})\right]s(t_{k},z_{k})\right\}\d t_{4}\d t_{3}\d t_{2}\d t_{1}\d z_{4}\d z_{3}\d z_{2}\d z_{1}\nonumber\\
&=&\epsilon^{4}\iiint_{[0,s)^{3}}\int_{[s,t)}\iiiint_{\mathbb{R}^{4}}\mathrm{E}\left\{\prod_{k=1}^{2}\left\{\im\left[W(t_{k},z_{k})\right]r(t_{k},z_{k})\right\}\prod_{k=3}^{4}\left\{\re\left[W(t_{k},z_{k})\right]s(t_{k},z_{k})\right\}\right\}\d t_{4}\d t_{3}\d t_{2}\d t_{1}\d z_{4}\d z_{3}\d z_{2}\d z_{1}\nonumber\\
&=&\epsilon^{4}\iiint_{[0,s)^{3}}\int_{[s,t)}\iiiint_{\mathbb{R}^{4}}\mathrm{E}\left\{\left\{\re\left[\underline{W(t_{3},z_{3})}\right]\prod_{k=1}^{2}\left\{\im\left[\underline{W(t_{k},z_{k})}\right]r(t_{k},z_{k})\right\}\prod_{k=3}^{4}s(t_{k},z_{k})\Bigg|\underline{\textbf{Q}_{k}}=\textbf{Q}_{k},k=1,2,3,4.\right\}\right\}\cdot\nonumber\\
&&\mathrm{E}\left\{\re\left[\underline{G(t_{4},z_{4})}e^{-j\varphi(t_{4},z_{4})}\right]\right\}\d t_{4}\d t_{3}\d t_{2}\d t_{1}\d z_{4}\d z_{3}\d z_{2}\d z_{1}\label{ch3.3.61}\\
&=& 0,\label{ch3.3.62}
\end{eqnarray}
where we have \eqref{ch3.3.61} because of Lemma \ref{ch3.l3.16} and the definition of \ito\ calculus. Equation \eqref{ch3.3.62} holds because of Lemma \ref{ch3.l3.16}. So \eqref{ch3.3.57} is proved.

The proof of \eqref{ch3.3.58} is similar to that of \eqref{ch3.3.57}, and hence is omitted.

To prove \eqref{ch3.3.59}, we consider 
\begin{eqnarray}
&&\mathrm{E}\left\{\nr(s)\nu_{I}(s)\left[\nr(t)-\nr(s)\right]\left[\nu_{I}(t)-\nu_{I}(s)\right]\right\}\nonumber\\
&=&\epsilon^{4}\mathrm{E}\iint_{[0,s)^{2}}\iint_{[s,t)^{2}}\iiiint_{\mathbb{R}^{4}}\prod_{k\in\{1,3\}}\left\{\im\left[W(t_{k},z_{k})\right]r(t_{k},z_{k})\right\}\prod_{k\in\{2,4\}}\left\{\re\left[W(t_{k},z_{k})\right]s(t_{k},z_{k})\right\}\nonumber\\
&&\d t_{4}\d t_{3}\d t_{2}\d t_{1}\d z_{4}\d z_{3}\d z_{2}\d z_{1}\nonumber\\
&=&\epsilon^{4}\iint_{[0,s)^{2}}\iint_{[s,t)^{2}}\iiiint_{\mathbb{R}^{4}}\mathrm{E}\Bigg\{\prod_{k\in\{1,3\}}\left\{\im\left[W(t_{k},z_{k})\right]r(t_{k},z_{k})\right\}\prod_{k\in\{2,4\}}\left\{\re\left[W(t_{k},z_{k})\right]s(t_{k},z_{k})\right\}\Bigg\}\nonumber\\
&&\d t_{4}\d t_{3}\d t_{2}\d t_{1}\d z_{4}\d z_{3}\d z_{2}\d z_{1},\label{ch3.3.64}
\end{eqnarray}
where we denote the integrand in \eqref{ch3.3.64} by $h(t_{1},t_{2},t_{3},t_{4},z_{1},z_{2},z_{3},z_{4})$, where $t_{k}\in(-\infty,\infty)$, $k\in\{1,2,3,4\}$, and $z_{m}\in[0,s)$, $m\in\{1,2\}$, and $z_{n}\in[s,t)$, $n\in\{3,4\}$. We divide the region of the integration (denoted by $S$) in \eqref{ch3.3.64} by two disjoint parts $S_{1}$ and $S_{2}$, where $S_{1}\triangleq[0,s)\times[0,s)\times\{z_{3}\neq z_{4}\in[s,t)\}\times\mathbb{R}^{4}$, and $S_{2}\triangleq[0,s)\times[0,s)\times\{z_{3}=z_{4}\in[s,t)\}\times\mathbb{R}^{4}$. So we have $S=S_{1}\cup S_{2}$, and $S_{1}\cap S_{2}=\emptyset$. Then in $S_{1}$, we assume $z_{3}<z_{4}$ without loss of generality. So the integrand in $S_{1}$ becomes 
\begin{eqnarray}
&&h(t_{1},t_{2},t_{3},t_{4},z_{1},z_{2},z_{3},z_{4})\nonumber\\
&=&\mathrm{E}\Bigg\{\mathrm{E}\Bigg\{\re\left[\underline{W(t_{2},z_{2})}\right]\prod_{k\in\{1,3\}}\left\{\im\left[\underline{W(t_{k},z_{k})}\right]r(t_{k},z_{k})\right\}\prod_{k\in\{2,4\}}s(t_{k},z_{k})\Bigg|\underline{\textbf{Q}_{k}}=\textbf{Q}_{k}, k=1,2,3,4.\Bigg\}\Bigg\}\cdot\nonumber\\
&&\mathrm{E}\left\{\re\left[\underline{G(t_{4},z_{4})}e^{-j\varphi(t_{4},z_{4})}\right]\right\}\label{ch3.3.65}\\
&=& 0,\label{ch3.3.66}
\end{eqnarray}
where \eqref{ch3.3.65} is obtained because of Lemma \ref{ch3.l3.16} and the definition of It\^{o} stochastic calculus. Equation \eqref{ch3.3.66} holds because of Lemma \ref{ch3.l3.16}. The integrand in $S_{2}$ is 
\begin{eqnarray}
&&h(t_{1},t_{2},t_{3},t_{4},z_{1},z_{2},z_{3},z_{4})\nonumber\\
&=&\mathrm{E}\Bigg\{\mathrm{E}\Bigg\{\im\left[\underline{W(t_{1},z_{1})}\right]\re\left[\underline{W(t_{2},z_{2})}\right]\prod_{k\in\{1,3\}}r(t_{k},z_{k})\prod_{k\in\{2,4\}}r(t_{k},z_{k})\Bigg|\underline{\textbf{Q}_{k}}=\textbf{Q}_{k}, k=1,2,3,4.\Bigg\}\Bigg\}\cdot\nonumber\\
&&\mathrm{E}\left\{\im\left[\underline{G(t_{3},z_{3})}e^{-j\varphi(t_{3},z_{3})}\right]\re\left[\underline{G(t_{4},z_{4})}e^{-j\varphi(t_{4},z_{4})}\right]\right\}\label{ch3.3.67}\\
&=& 0,\label{ch3.3.68}
\end{eqnarray}
where \eqref{ch3.3.67} is obtained because of Lemma \ref{ch3.l3.16} and the definition of It\^{o} stochastic calculus. Equation \eqref{ch3.3.68} holds because of Lemma \ref{ch3.l3.16}. So according to \eqref{ch3.3.66}--\eqref{ch3.3.68}, we have $h(t_{1},t_{2},t_{3},t_{4},z_{1},z_{2},z_{3},z_{4})=0$ in the whole integration region. So we have \eqref{ch3.3.59}.

2) Since $s\leq t$, we have 
\begin{eqnarray}
&&\mathrm{E}\left[\nr(s)\nu_{I}(s)\nr(t)\nu_{I}(t)\right]\nonumber\\
&=&\mathrm{E}\left\{\nr(s)\nu_{I}(s)\left[\nr(t)-\nr(s)+\nr(s)\right]\left[\nu_{I}(t)-\nu_{I}(s)+\nu_{I}(s)\right]\right\}\nonumber\\
&=& \mathrm{E}\left\{\nr(s)^{2}\nu_{I}(s)\left[\nu_{I}(t)-\nu_{I}(s)\right]\right\}+\mathrm{E}\left\{\nr(s)\nu_{I}(s)^{2}\left[\nr(t)-\nr(s)\right]\right\}+\nonumber\\
&&\mathrm{E}\left\{\nr(s)\nu_{I}(s)\left[\nr(t)-\nr(s)\right]\left[\nu_{I}(t)-\nu_{I}(s)\right]\right\}+\mathrm{E}\left[\nr(s)^{2}\nu_{I}(s)^{2}\right]\nonumber\\
&=&\mathrm{E}\left[\nr(s)^{2}\nu_{I}(s)^{2}\right],\label{ch3.3.69}
\end{eqnarray}
which is obtained by the first part of this lemma. Consider $\mathrm{E}\left[\nr(s)^{2}\nu_{I}(s)^{2}\right]$ 
\begin{eqnarray}
&&\mathrm{E}\left[\nr(s)^{2}\nu_{I}(s)^{2}\right]\nonumber\\
&=&\epsilon^{4}\mathrm{E}\iiiint_{[0,s)^{4}}\iiiint_{\mathbb{R}^{4}}\prod_{k=1}^{2}\left\{\im\left[W(t_{k},z_{k})\right]r(t_{k},z_{k})\right\}\prod_{k=3}^{4}\left\{\re\left[W(t_{k},z_{k})\right]s(t_{k},z_{k})\right\}\nonumber\\
&&\d t_{4}\d t_{3}\d t_{2}\d t_{1}\d z_{4}\d z_{3}\d z_{2}\d z_{1}\nonumber\\
&=&\epsilon^{4}\iiiint_{[0,s)^{4}}\iiiint_{\mathbb{R}^{4}}\mathrm{E}\Bigg\{\prod_{k=1}^{2}\left\{\im\left[W(t_{k},z_{k})\right]r(t_{k},z_{k})\right\}\prod_{k=3}^{4}\left\{\re\left[W(t_{k},z_{k})\right]s(t_{k},z_{k})\right\}\Bigg\}\nonumber\\
&&\d t_{4}\d t_{3}\d t_{2}\d t_{1}\d z_{4}\d z_{3}\d z_{2}\d z_{1}.\label{ch3.3.70}
\end{eqnarray}
Denote the integrand in \eqref{ch3.3.70} by $h(t_{1},t_{2},t_{3},t_{4},z_{1},z_{2},z_{3},z_{4})$. Denote
\begin{equation}
f(z_{1},z_{2},z_{3},z_{4})\triangleq\int_{-\infty}^{\infty}\int_{-\infty}^{\infty}\int_{-\infty}^{\infty}\int_{-\infty}^{\infty}h(t_{1},t_{2},t_{3},t_{4},z_{1},z_{2},z_{3},z_{4})\d t_{4}\d t_{3}\d t_{2}\d t_{1}.\label{ch3.3.71}
\end{equation} 
We choose a permutation of $(1,2,3,4)$, and denote it by $(k_{1},k_{2},k_{3},k_{4})$. Without loss of generality, we assume $z_{k_{1}}\leq z_{k_{2}}\leq z_{k_{3}}\leq z_{k_{4}}$. Then we divide the region 
\begin{equation}
S'\triangleq\{(t_{k_{1}},t_{k_{2}},t_{k_{3}},t_{k_{4}},z_{k_{1}},z_{k_{2}},z_{k_{3}},z_{k_{4}})\in\mathbb{R}^{4}\times[0,s)^{4}:\ z_{k_{1}}\leq z_{k_{2}}\leq z_{k_{3}}\leq z_{k_{4}}\}\nonumber
\end{equation}
by two disjoint parts denoted respectively by $S_{3}$ and $S_{4}$, where
\begin{equation}
S_{3}\triangleq\{(t_{k_{1}},t_{k_{2}},t_{k_{3}},t_{k_{4}},z_{k_{1}},z_{k_{2}},z_{k_{3}},z_{k_{4}})\in\mathbb{R}^{4}\times[0,s)^{4}:\ z_{k_{1}}\leq z_{k_{2}}=z_{k_{3}}\leq z_{k_{4}}\},\nonumber
\end{equation}
and
\begin{equation}
S_{4}\triangleq\{(t_{k_{1}},t_{k_{2}},t_{k_{3}},t_{k_{4}},z_{k_{1}},z_{k_{2}},z_{k_{3}},z_{k_{4}})\in\mathbb{R}^{4}\times[0,s)^{4}:\ z_{k_{1}}\leq z_{k_{2}}<z_{k_{3}}\leq z_{k_{4}}\}.\nonumber\\
\end{equation}
So we have $S'=S_{3}\cup S_{4}$. Then we consider the integrand $$h(t_{k_{1}},t_{k_{2}},t_{k_{3}},t_{k_{4}},z_{k_{1}},z_{k_{2}},z_{k_{3}},z_{k_{4}})$$
in $S_{3}$ and $S_{4}$.

In the region $S_{3}$, we have 
\begin{itemize}
{
\item When $z_{k_{3}}<z_{k_{4}}$, similar to the discussion in obtaining \eqref{ch3.3.62}, we have $h=0$;
\item When $z_{k_{3}}=z_{k_{4}}$, we have
  \begin{itemize}
  {
   \item When $z_{k_{1}}<z_{k_{2}}$, we have $z_{k_{1}}<z_{k_{2}}=z_{k_{3}}=z_{k_{4}}$. Without loss of generality, we assume the dummy variables $z_{k_{1}}$ and $z_{k_{2}}$ correspond to $\nr(z)$. So the integrand is 
\begin{eqnarray}
   &&h(t_{k_{1}},t_{k_{2}},t_{k_{3}},t_{k_{4}},z_{k_{1}},z_{k_{2}},z_{k_{3}},z_{k_{4}})\nonumber\\
   &=& \mathrm{E}\left\{\mathrm{E}\left\{\im\left[\underline{W(t_{k_{1}},z_{k_{1}})}\right]\prod_{l=1}^{2}r(t_{k_{l}},z_{k_{l}})\prod_{l=3}^{4}s(t_{k_{l}},z_{k_{l}})\Bigg|\underline{\textbf{Q}_{k_{l}}}=\textbf{Q}_{k_{l}}, l=1,2,3,4\right\}\right\}\cdot\nonumber\\
   &&\mathrm{E}\left\{\im\left[\underline{W(t_{k_{2}},z_{k_{2}})}\right]\prod_{l=3}^{4}\re\left[\underline{W(t_{k_{l}},z_{k_{l}})}\right]\Bigg|\underline{\textbf{Q}_{k_{l}}}=\textbf{Q}_{k_{l}}, l=1,2,3,4\right\}\label{ch3.3.72}\\
 &=& \mathrm{E}\left\{\mathrm{E}\left\{\im\left[\underline{W(t_{k_{1}},z_{k_{1}})}\right]\prod_{l=1}^{2}r(t_{k_{l}},z_{k_{l}})\prod_{l=3}^{4}s(t_{k_{l}},z_{k_{l}})\Bigg|\underline{\textbf{Q}_{k_{l}}}=\textbf{Q}_{k_{l}}, l=1,2,3,4\right\}\right\}\cdot\nonumber\\
 &&\mathrm{E}\left\{\im\left[\underline{G(t_{k_{2}},z_{k_{2}})}e^{-j\varphi(t_{k_{2}},z_{k_{2}})}\right]\right\}\mathrm{E}\left\{\prod_{l=3}^{4}\re\left[\underline{G(t_{k_{l}},z_{k_{l}})}e^{-j\varphi(t_{k_{l}},z_{k_{l}})}\right]\right\}\label{ch3.3.73}\\
&=& 0,\label{ch3.3.74}
   \end{eqnarray}
   where \eqref{ch3.3.72} holds because of Lemma \ref{ch3.l3.16} and the definition of It\^{o} stochastic calculus. Equation \eqref{ch3.3.73} holds because of Lemma \ref{ch3.l3.16}.
   \item When $z_{k_{1}}=z_{k_{2}}$, we have $z_{k_{1}}=z_{k_{2}}=z_{k_{3}}=z_{k_{4}}$, which means that $z_{1}=z_{2}=z_{3}=z_{4}$. Without loss of generality, we assume the dummy variables $z_{1}$ and $z_{2}$ correspond to $\nr(z)$. Then we have 
\begin{eqnarray}
   &&h(t_{1},t_{2},t_{3},t_{4},z_{1},z_{2},z_{3},z_{4})\nonumber\\
   &=&\mathrm{E}\left\{\prod_{k=1}^{2}\im\left[\underline{G(t_{k},z_{k})}e^{-j\varphi(t_{k},z_{k})}\right]\prod_{k=3}^{4}\re\left[\underline{G(t_{k},z_{k})}e^{-j\varphi(t_{k},z_{k})}\right]\right\}\mathrm{E}\left\{\prod_{k=1}^{2}\underline{r(t_{k},z_{k})}\prod_{k=3}^{4}\underline{s(t_{k},z_{k})}\right\}\nonumber\\
&=& \frac{1}{4}\left[\prod_{k\in\{1,3\}}\delta(t_{k}-t_{k+1})\delta(z_{k}-z_{k+1})\cos[\varphi(t_{k},z_{1})-\varphi(t_{k+1},z_{1})]\right]\nonumber\\
&&\mathrm{E}\left\{\prod_{k=1}^{2}\underline{r(t_{k},z_{k})}\prod_{k=3}^{4}\underline{s(t_{k},z_{k})}\right\},\label{ch3.3.77}
   \end{eqnarray}
   which is obtained using Lemma \ref{ch3.l3.12}. Lemma \ref{ch3.l3.12} can be applied here because of Lemma \ref{ch3.l3.16}. Using \eqref{ch3.3.71}, we have 
\begin{eqnarray}
   &&4f(z_{1},z_{2},z_{3},z_{4})\nonumber\\
   &=&\iiiint_{\mathbb{R}^{4}}\left[\prod_{k\in\{1,3\}}\delta(t_{k}-t_{k+1})\delta(z_{k}-z_{k+1})\cos[\varphi(t_{k},z_{1})-\varphi(t_{k+1},z_{1})]\right]\nonumber\\
   &&\mathrm{E}\left\{\prod_{k=1}^{2}r(t_{k},z_{k})\prod_{k=3}^{4}s(t_{k},z_{k})\right\}\d t_{4}\d t_{3}\d t_{2}\d t_{1}\nonumber\\
   &=&\mathrm{E}\Bigg\{\iiiint_{\mathbb{R}^{4}}\left[\prod_{k\in\{1,3\}}\delta(t_{k}-t_{k+1})\delta(z_{k}-z_{k+1})\cos[\varphi(t_{k},z_{1})-\varphi(t_{k+1},z_{1})]\right]\nonumber\\
   &&\prod_{k=1}^{2}r(t_{k},z_{k})\prod_{k=3}^{4}s(t_{k},z_{k})\d t_{4}\d t_{3}\d t_{2}\d t_{1}\Bigg\}\nonumber\\
   &=&\mathrm{E}\left\{\int_{-\infty}^{\infty}r(t_{1},z_{1})^{2}\d t_{1}\int_{-\infty}^{\infty}s(t_{3},z_{3})^{2}\d t_{3}\right\}\delta(z_{1}-z_{2})\delta(z_{3}-z_{4})\nonumber\\
   &=& \frac{1}{3}\mathrm{E}[\beta(z_{1})\beta(z_{3})]\delta(z_{1}-z_{2})\delta(z_{3}-z_{4}),\label{ch3.3.78}
   \end{eqnarray}
where $z_{1}=z_{2}=z_{3}=z_{4}$. 
  }
  \end{itemize}
So in summary, we have 
\begin{eqnarray}
\idotsint_{S_{3}}f(z_{1},z_{2},z_{3},z_{4})\d\sigma_{4}&=&\idotsint_{S_{3}}\frac{1}{12}\mathrm{E}[\beta(z_{k_{1}})\beta(z_{k_{3}})]\delta(z_{k_{1}}-z_{k_{2}})\delta(z_{k_{3}}-z_{k_{4}})\d\sigma_{4}\nonumber\\
&=&\iint_{\{z_{k_{2}}=z_{k_{3}}\}}\frac{1}{12}\mathrm{E}[\beta(z_{k_{1}})\beta(z_{k_{3}})]\d\sigma_{2}=0.\label{ch3.3.102}
\end{eqnarray}
So in the region of the union of all possible $S_{3}$, the integrand in \eqref{ch3.3.102} is a non-zero real function only when $z_{k_{2}}=z_{k_{3}}$. So according to the definition of multiple integration, the integration of $h(t_{1},t_{2},t_{3},t_{4},z_{1},z_{2},z_{3},z_{4})$ over the region of the union of all possible $S_{3}$ is zero.
}
\end{itemize}

In region $S_{4}$, we have
\begin{itemize}
{
\item When $z_{k_{3}}<z_{k_{4}}$, similar to the discussion in obtaining \eqref{ch3.3.62}, we have $h=0$;
\item When $z_{k_{3}}=z_{k_{4}}$, similar to the discussion in obtaining \eqref{ch3.3.68}, the integrand $h$ can be non-zero only when the dummy variables $z_{k_{3}}$ and $z_{k_{4}}$ correspond simultaneously to either both $\nr(z)$ or both $\nu_{I}(z)$. Without loss of generality, we assume the dummy variables $z_{k_{3}}$ and $z_{k_{4}}$ correspond simultaneously to $\nu_{I}(z)$ at first. So the integrand is 
\begin{eqnarray}
   &&h(t_{k_{1}},t_{k_{2}},t_{k_{3}},t_{k_{4}},z_{k_{1}},z_{k_{2}},z_{k_{3}},z_{k_{4}})\nonumber\\
   &=& \mathrm{E}\Bigg\{\mathrm{E}\Bigg\{\prod_{l=1}^{2}\left\{\im\left[\underline{W(t_{k_{l}},z_{k_{l}})}\right]r(t_{k_{l}},z_{k_{l}})\right\}\prod_{l=3}^{4}s(t_{k_{l}},z_{k_{l}})\Bigg|\underline{\textbf{Q}_{k_{l}}}=\textbf{Q}_{k_{l}}, l=1,2,3,4\Bigg\}\Bigg\}\cdot\nonumber\\
  && \mathrm{E}\left\{\prod_{l=3}^{4}\re\left[\underline{G(t_{k_{l}},z_{k_{l}})}e^{-j\varphi(t_{k_{l}},z_{k_{l}})}\right]\right\}\label{ch3.3.79}\\
 &=& \frac{1}{2}\delta(t_{k_{3}}-t_{k_{4}})\delta(0)\cos[\varphi(t_{k_{3}},z_{k_{3}})-\varphi(t_{k_{4}},z_{k_{3}})]\cdot\nonumber\\
 &&\mathrm{E}\Bigg\{\mathrm{E}\Bigg\{\prod_{l=1}^{2}\left\{\im\left[\underline{W(t_{k_{l}},z_{k_{l}})}\right]r(t_{k_{l}},z_{k_{l}})\right\}\prod_{l=3}^{4}s(t_{k_{l}},z_{k_{l}})\Bigg|\underline{\textbf{Q}_{k_{l}}}=\textbf{Q}_{k_{l}}, l=1,2,3,4\Bigg\}\Bigg\}\label{ch3.3.80}
   \end{eqnarray} 
   So we have 
\begin{eqnarray}
   &&f(z_{k_{1}},z_{k_{2}},z_{k_{3}},z_{k_{4}})\nonumber\\
   &=&\iiiint_{\mathbb{R}^{4}}h(t_{k_{1}},t_{k_{2}},t_{k_{3}},t_{k_{4}},z_{k_{1}},z_{k_{2}},z_{k_{3}},z_{k_{4}})\d t_{k_{4}}\d t_{k_{3}}\d t_{k_{2}}\d t_{k_{1}}\nonumber\\
&=&\frac{1}{2}\mathrm{E}\iiiint_{\mathbb{R}^{4}}\delta(t_{k_{3}}-t_{k_{4}})\delta(0)\cos[\varphi(t_{k_{3}},z_{k_{3}})-\varphi(t_{k_{4}},z_{k_{3}})]\prod_{l=1}^{2}\left\{\im\left[W(t_{k_{l}},z_{k_{l}})\right]r(t_{k_{l}},z_{k_{l}})\right\}\prod_{l=3}^{4}s(t_{k_{l}},z_{k_{l}})\nonumber\\
&&\d t_{k_{4}}\d t_{k_{3}}\d t_{k_{2}}\d t_{k_{1}}\nonumber\\
&=&\frac{1}{2}\mathrm{E}\iint_{\mathbb{R}^{2}}\prod_{l=1}^{2}\left\{\im\left[W(t_{k_{l}},z_{k_{l}})\right]r(t_{k_{l}},z_{k_{l}})\right\}\beta(z_{k_{3}})\delta(0)\d t_{k_{2}}\d t_{k_{1}}\nonumber\\
&=&\frac{1}{2}\mathrm{E}\iint_{\mathbb{R}^{2}}\prod_{l=1}^{2}\left\{\im\left[W(t_{k_{l}},z_{k_{l}})\right]r(t_{k_{l}},z_{k_{l}})\right\}\left\{\beta(z_{k_{2}})+\left[\beta(z_{k_{3}})-\beta(z_{k_{2}})\right]\right\}\delta(0)\d t_{k_{2}}\d t_{k_{1}}\nonumber\\
&=&\frac{1}{2}\mathrm{E}\iint_{\mathbb{R}^{2}}\prod_{l=1}^{2}\left\{\im\left[W(t_{k_{l}},z_{k_{l}})\right]r(t_{k_{l}},z_{k_{l}})\right\}\left[\beta(z_{k_{2}})+\frac{1}{2}\epsilon^{2}(z_{k_{3}}-z_{k_{2}})\right]\delta(0)\d t_{k_{2}}\d t_{k_{1}}+\nonumber\\
&&\frac{\epsilon}{2}\mathrm{E}\int_{(z_{k_{2}},z_{k_{3}})}\iiint_{\mathbb{R}^{3}}\left\{\prod_{l=1}^{2}\im\left[W(t_{k_{l}},z_{k_{l}})\right]r(t_{k_{l}},z_{k_{l}})\right\}\re\left[W(t_{k_{5}},z_{k_{5}})\right]s(t_{k_{5}},z_{k_{5}})\delta(0)\d t_{k_{5}}\d t_{k_{2}}\d t_{k_{1}}\d z_{k_{5}}\nonumber\\
&=&\frac{1}{2}\mathrm{E}\iint_{\mathbb{R}^{2}}\prod_{l=1}^{2}\left\{\im\left[\underline{W(t_{k_{l}},z_{k_{l}})}\right]\underline{r(t_{k_{l}},z_{k_{l}})}\right\}\left[\underline{\beta(z_{k_{2}})}+\frac{1}{2}\epsilon^{2}(z_{k_{3}}-z_{k_{2}})\right]\delta(0)\d t_{k_{2}}\d t_{k_{1}}+\nonumber\\
&&\frac{\epsilon}{2}\int_{(z_{k_{2}},z_{k_{3}})}\iiint_{\mathbb{R}^{3}}\mathrm{E}\left\{\mathrm{E}\left\{\left\{\prod_{l=1}^{2}\im\left[\underline{W(t_{k_{l}},z_{k_{l}})}\right]r(t_{k_{l}},z_{k_{l}})\right\}s(t_{k_{5}},z_{k_{5}})\Bigg|\underline{\textbf{Q}_{k_{l}}}=\textbf{Q}_{k_{l}}, l=1,2,5\right\}\right\}\cdot\nonumber\\
&&\mathrm{E}\left\{\re\left[\underline{G(t_{k_{5}},z_{k_{5}})}e^{-j\varphi(t_{k_{5}},z_{k_{5}})}\right]\right\}\delta(0)\d t_{k_{5}}\d t_{k_{2}}\d t_{k_{1}}\d z_{k_{5}}\nonumber\\
&=&\frac{1}{2}\mathrm{E}\iint_{\mathbb{R}^{2}}\prod_{l=1}^{2}\left\{\im\left[W(t_{k_{l}},z_{k_{l}})\right]r(t_{k_{l}},z_{k_{l}})\right\}\left[\beta(z_{k_{2}})+\frac{1}{2}\epsilon^{2}(z_{k_{3}}-z_{k_{2}})\right]\delta(0)\d t_{k_{2}}\d t_{k_{1}}+0.\label{ch3.3.81}
   \end{eqnarray}
   Furthermore, we have
   \begin{itemize}
   {
    \item When $z_{k_{1}}<z_{k_{2}}$, we have $z_{k_{1}}<z_{k_{2}}<z_{k_{3}}=z_{k_{4}}$. Similar to the discussion in obtaining \eqref{ch3.3.62}, we have $f=0$;
    \item When $z_{k_{1}}=z_{k_{2}}$, we have $z_{k_{1}}=z_{k_{2}}<z_{k_{3}}=z_{k_{4}}$. Continue the derivation from \eqref{ch3.3.81}, we have 
\begin{eqnarray}
   &&f(z_{k_{1}},z_{k_{2}},z_{k_{3}},z_{k_{4}})\nonumber\\
   &=&\frac{1}{4}\iint_{\mathbb{R}^{2}}\mathrm{E}\Bigg\{\left[\prod_{l=1}^{2}r(t_{k_{l}},z_{k_{2}})\right]\left[\beta(z_{k_{2}})+\frac{1}{2}\epsilon^{2}(z_{k_{3}}-z_{k_{2}})\right]\delta(0)\Bigg\}\delta(t_{k_{1}}-t_{k_{2}})\delta(0)\nonumber\\
   &&\cos[\varphi(t_{k_{1}},z_{k_{2}})-\varphi(t_{k_{2}},z_{k_{2}})]\Big.\d t_{k_{2}}\d t_{k_{1}}\nonumber\\
   &=&\frac{1}{4}\mathrm{E}\int_{-\infty}^{\infty}\left[\beta(z_{k_{2}})+\frac{1}{2}\epsilon^{2}(z_{k_{3}}-z_{k_{2}})\right]r(t_{k_{2}},z_{k_{2}})^{2}\delta(0)\delta(0)\d t_{k_{2}}\nonumber\\
   &=&\frac{1}{12}\left\{\mathrm{E}\left[\beta(z_{k_{2}})^{2}\right]+\frac{1}{2}\epsilon^{2}(z_{k_{3}}-z_{k_{2}})\mathrm{E}\beta(z_{k_{2}})\right\}\delta(0)\delta(0).\label{ch3.3.82}
   \end{eqnarray}
   }
   \end{itemize}
If we assume the dummy variables $z_{k_{3}}$ and $z_{k_{4}}$ correspond simultaneously to $\nr(z)$, similar to the derivation \eqref{ch3.3.79}--\eqref{ch3.3.82}, we also have
   \begin{itemize}
   {
    \item When $z_{k_{1}}<z_{k_{2}}$, we have $f=0$;
    \item When $z_{k_{1}}=z_{k_{2}}$, we have
    \begin{equation}
   f(z_{k_{1}},z_{k_{2}},z_{k_{3}},z_{k_{4}})=\frac{1}{12}\bigg\{\mathrm{E}\left[\beta(z_{k_{2}})^{2}\right]+\frac{1}{2}\epsilon^{2}(z_{k_{3}}-z_{k_{2}})\mathrm{E}\beta(z_{k_{2}})\bigg\}\delta(0)\delta(0).\label{ch3.3.83}
   \end{equation}
   }
   \end{itemize}
}
\end{itemize}
So in the region $S_{4}$, we have
\begin{equation}
   f(z_{k_{1}},z_{k_{2}},z_{k_{3}},z_{k_{4}})=\frac{1}{12}\bigg\{\mathrm{E}\left[\beta(z_{k_{2}})^{2}\right]+\frac{1}{2}\epsilon^{2}(z_{k_{3}}-z_{k_{2}})\mathrm{E}\beta(z_{k_{2}})\bigg\}\delta(z_{k_{3}}-z_{k_{4}})\delta(z_{k_{1}}-z_{k_{2}}).\label{ch3.3.103}
   \end{equation}
So in summary, in region of the union of all possible $S_{4}$ (denoted by $S_{4}'$), we have 
\begin{eqnarray}
&&\iiiint_{S_{4}'}\int_{-\infty}^{\infty}\int_{-\infty}^{\infty}\int_{-\infty}^{\infty}\int_{-\infty}^{\infty}h(t_{1},t_{2},t_{3},t_{4},z_{1},z_{2},z_{3},z_{4})\d t_{4}\d t_{3}\d t_{2}\d t_{1}\d S_{4}'\nonumber\\
&=&\iiiint_{S_{4}'}f(z_{k_{1}},z_{k_{2}},z_{k_{3}},z_{k_{4}})\d S_{4}\nonumber\\
&=&2\iint_{\{z_{k_{2}}<z_{k_{3}}\}}\frac{1}{12}\left\{\mathrm{E}\left[\beta(z_{k_{2}})^{2}\right]+\frac{1}{2}\epsilon^{2}(z_{k_{3}}-z_{k_{2}})\mathrm{E}\beta(z_{k_{2}})\right\}\d\sigma\nonumber\\
&=&\frac{1}{6}\iint_{\{z_{k_{2}}<z_{k_{3}}\}}\mathrm{E}\left[\beta(0)^{2}\right]+\frac{3}{2}\epsilon^{2}z_{k_{2}}\mathrm{E}\beta(0)+\frac{3}{8}\epsilon^{4}z_{k_{2}}^{2}+\frac{1}{2}\epsilon^{2}\left(z_{k_{3}}-z_{k_{2}}\right)\left(\mathrm{E}\beta(0)+\frac{1}{2}\epsilon^{2}z_{k_{2}}\right)\d\sigma\nonumber\\
&=&\frac{1}{12}s^{2}\mathrm{E}\left[\beta(0)^{2}\right]+\frac{1}{18}\epsilon^{2}s^{3}\mathrm{E}\beta(0)+\frac{1}{144}\epsilon^{4}s^{4}.\label{ch3.3.84}
\end{eqnarray}
As a result, for any $s\leq t$,
\begin{equation}
\mathrm{E}\left[\nr(s)\nu_{I}(s)\nr(t)\nu_{I}(t)\right]=\mathrm{E}\left[\nr(s)^{2}\nu_{I}(s)^{2}\right]=\frac{1}{12}\epsilon^{4}s^{2}\mathrm{E}\left[\beta(0)^{2}\right]+\frac{1}{18}\epsilon^{6}s^{3}\mathrm{E}\beta(0)+\frac{1}{144}\epsilon^{8}s^{4}.\nonumber
\end{equation}

3) We have 
\begin{eqnarray}
&&\mathrm{E}\left[\underline{\nr(s)}^{2}\underline{\nu_{I}(s)}\Big|\underline{\zeta(0)}=\zeta(0)\right]\nonumber\\
&=&\epsilon^{3}\mathrm{E}\Bigg\{\iiint_{[0,s)^{3}}\iiint_{\mathbb{R}^{3}}\left\{\prod_{k=1}^{2}\im\left[\underline{W(t_{k},z_{k})}\right]\underline{r(t_{k},z_{k})}\right\}\underline{s(t_{3},z_{3})}\re\left[\underline{W(t_{3},z_{3})}\right]\d t_{1}\d t_{2}\d t_{3}\d z_{1}\d z_{2}\d z_{3}\Bigg|\underline{\zeta(0)}=\zeta(0)\Bigg\}\nonumber\\
&=&\epsilon^{3}\iiint_{[0,s)^{3}}\iiint_{\mathbb{R}^{3}}\mathrm{E}\Bigg\{\mathrm{E}\left\{\prod_{k=1}^{2}\im\left[\underline{W(t_{k},z_{k})}\right]\underline{r(t_{k},z_{k})}\right\}\underline{s(t_{3},z_{3})}\re\left[\underline{W(t_{3},z_{3})}\right]\Bigg|\underline{\zeta(0)}=\zeta(0)\Bigg\}\nonumber\\
&&\d t_{1}\d t_{2}\d t_{3}\d z_{1}\d z_{2}\d z_{3}\label{ch3.a.104}.
\end{eqnarray}
where we denote the integrand in \eqref{ch3.a.104} by $h(t_{1},t_{2},t_{3},z_{1},z_{2},z_{3})$. Denote
\begin{equation}
f(z_{1},z_{2},z_{3})\triangleq\int_{-\infty}^{\infty}\int_{-\infty}^{\infty}\int_{-\infty}^{\infty}h(t_{1},t_{2},t_{3},z_{1},z_{2},z_{3})\d t_{1}\d t_{2}\d t_{3}.\label{ch3.a.105}
\end{equation}

Without loss of generality, we assume $z_{1}\leq z_{2}$. We divide the integration region into two disjoint parts $\{z_{3}\leq z_{1}\}$ and $\{z_{3}>z_{1}\}$, and consider the integrand in these regions.

In the region $\{z_{3}>z_{1}\}$,
\begin{itemize}
{
\item If $z_{2}\neq z_{3}$, then similar to the proof of \eqref{ch3.3.62}, we have $h=0$;
\item If $z_{3}=z_{2}$, then 
\begin{itemize}
{
\item if $z_{1}<z_{2}$, similar to the proof of \eqref{ch3.3.68}, we have $h=0$;
\item if $z_{1}=z_{2}$, using Lemma \ref{ch3.l3.12}, we have $h=0$.
}
\end{itemize} 
}
\end{itemize}
In the region $\{z_{3}\leq z_{1}\}$, we have $z_{3}\leq z_{1}\leq z_{2}$. Then
\begin{itemize}
{
\item If $z_{3}\leq z_{1}< z_{2}$, then similar to the proof of \eqref{ch3.3.62}, we have $h=0$;
\item If $z_{3}< z_{1}= z_{2}$, then similar to the proof of \eqref{ch3.3.80}, we have 
\begin{multline}
h(t_{1},t_{2},t_{3},z_{1},z_{2},z_{3})\\
=\frac{1}{2}\mathrm{E}\Bigg\{\mathrm{E}\Bigg\{\left[\prod_{k=1}^{2}\underline{r(t_{k},z_{k})}\right]\underline{s(t_{3},z_{3})}\re\left[\underline{W(t_{3},z_{3})}\right]\Bigg|\underline{\zeta(0)}=\zeta(0), \underline{\textbf{Q}_{k}}=\textbf{Q}_{k}, k=1,2,3\Bigg\}\Bigg\}\delta(t_{1}-t_{2})\delta(0).\label{ch3.a.111}
\end{multline}
Then similar to the proof of \eqref{ch3.3.81}, we have 
\begin{eqnarray}
f(z_{1},z_{2},z_{3})&=&\frac{1}{6}\mathrm{E}\int_{-\infty}^{\infty}\re\left[\underline{W(t_{3},z_{3})}\right]\underline{s(t_{3},z_{3})}\underline{\beta(z_{1})}\delta(0)\d t_{3}\nonumber\\
&=&\frac{1}{6}\mathrm{E}\int_{-\infty}^{\infty}\re\left[\underline{W(t_{3},z_{3})}\right]\underline{s(t_{3},z_{3})}\left[\beta(z_{3})+\frac{1}{2}\epsilon^{2}(z_{1}-z_{3})\right]\delta(0)\d t_{3}+0\nonumber\\
&=&0.\label{ch3.b.1}
\end{eqnarray}
\item If $z_{3}= z_{1}= z_{2}$, then 
\begin{eqnarray}
h(t_{1},t_{2},t_{3},z_{1},z_{2},z_{3})&=&h(t_{1},t_{2},t_{3},z_{1},z_{1},z_{1})\nonumber\\
&=&\iiint_{\mathbb{R}^{3}}\mathrm{E}\left\{\underline{s(t_{3},z_{1})}\prod_{k=1}^{2}\underline{r(t_{k},z_{1})}\right\}\mathrm{E}\left\{\re\left[\underline{W(t_{3},z_{1})}\right]\prod_{k=1}^{2}\im\left[\underline{W(t_{k},z_{1})}\right]\right\}\d t_{3}\d t_{2}\d t_{1}\nonumber\\
&=&0.\label{ch3.a.112}
\end{eqnarray}
because of Lemma \ref{ch3.l3.16} and Lemma \ref{ch3.l3.12}.
}
\end{itemize}

In summary, we have
\begin{equation}\label{ch3.a.113}
{f(z_{1},z_{2},z_{3})=0,}
\end{equation}
for any $(z_{1},z_{2},z_{3})\in[0,s]^{3}$. So
$$\mathrm{E}\left[\nr(s)^{2}\nu_{I}(s)\Big|\underline{\zeta(0)}=\zeta(0)\right]=0.$$

4) Similar to \eqref{ch3.3.69}, equations \eqref{ch3.3.121}--\eqref{ch3.3.122} can be proved by further noticing 
\begin{eqnarray}
\ni(s)\ni(t)&=&\ni(s)^{2}+\ni(s)\left[\ni(t)-\ni(s)\right]\nonumber\\
&=& \ni(s)^{2}+\frac{1}{4}\epsilon^{4}s(t-s)+\frac{1}{2}\epsilon^{2}s\left[\nu_{I}(t)-\nu_{I}(s)\right]+\nonumber\\
&&\frac{1}{2}\epsilon^{2}(t-s)\nu_{I}(s)+\nu_{I}(s)\left[\nu_{I}(t)-\nu_{I}(s)\right].\label{ch3.3.123}
\end{eqnarray}

Combine \eqref{ch3.3.121}--\eqref{ch3.3.122} and \eqref{ch3.3.95}--\eqref{ch3.3.93}, it is easy to see that the integrals of the auto-correlation functions of $\nr(z)$ and $\ni(z)$ conditioned on an eigenvalue $\zeta(0)=\alpha(0)+\beta(0)$ are both finite. As a result, the noise processes $\nr(z)$ and $\ni(z)$ are both mean-square integrable conditioned on an eigenvalue $\zeta(0)=\alpha(0)+\beta(0)$ over the interval $[0,\length]$.

5) Assume $s<t$ without loss of generality, then we have  
\begin{eqnarray}
\mathrm{E}\left[\nr(s)\ni(s)\nr(t)\ni(t)\right]&=& \mathrm{E}\left\{\nr(s)\left[\frac{1}{2}\epsilon^{2}s+\nu_{I}(s)\right]\nr(t)\left[\frac{1}{2}\epsilon^{2}t+\nu_{I}(t)\right]\right\}\nonumber\\
&=& \mathrm{E}\left[\nr(s)\nu_{I}(s)\nr(t)\nu_{I}(t)\right]+\frac{1}{2}\epsilon^{2}s\mathrm{E}\left[\nr(s)\nr(t)\nu_{I}(t)\right]+\nonumber\\
&&\frac{1}{2}\epsilon^{2}t\mathrm{E}\left[\nr(s)\nu_{I}(s)\nr(t)\right]+\frac{1}{4}\epsilon^{4}st\mathrm{E}\left[\nr(s)\nr(t)\right].\label{ch3.3.104}
\end{eqnarray}

We calculate $\mathrm{E}\left[\nr(s)\nr(t)\nu_{I}(t)\right]$ 
\begin{eqnarray}
&&\mathrm{E}\left[\nr(s)\nr(t)\nu_{I}(t)\right]\nonumber\\
&=&\mathrm{E}\Big\{\nr(s)\left[\nr(s)+\nr(t)-\nr(s)\right]\left[\nu_{I}(s)+\nu_{I}(t)-\nu_{I}(s)\right]\Big\}\nonumber\\
&=&\mathrm{E}\left[\nr(s)^{2}\nu_{I}(s)\right]+\mathrm{E}\left\{\nr(s)\left[\nr(t)-\nr(s)\right]\nu_{I}(s)\right\}+\mathrm{E}\left\{\nr(s)^{2}\Big[\nu_{I}(t)-\nu_{I}(s)\Big]\right\}+\nonumber\\
&&\mathrm{E}\left\{\nr(s)\Big[\nr(t)-\nr(s)\Big]\Big[\nu_{I}(t)-\nu_{I}(s)\Big]\right\}\nonumber\\
&=&0+0+0+0=0.\label{ch3.3.106}
\end{eqnarray}
Similarly, we have
\begin{multline}
\mathrm{E}\left[\nr(s)\nu_{I}(s)\nr(t)\right]=\mathrm{E}\left[\nr(s)^{2}\nu_{I}(s)\right]+\mathrm{E}\left\{\nr(s)\left[\nr(t)-\nr(s)\right]\nu_{I}(s)\right\}=0+0=0.\label{ch3.3.107}
\end{multline}

Consider
\begin{eqnarray}
\mathrm{E}\left[\nr(s)\nr(t)\right]&=&\mathrm{E}\left[\nr(s)^{2}\right]+\mathrm{E}\left\{\nr(s)\left[\nr(t)-\nr(s)\right]\right\}\nonumber\\
&=&\frac{1}{6}\epsilon^{2}s\mathrm{E}\beta(0)+\frac{1}{24}\epsilon^{4}s^{2}.\label{ch3.3.108}
\end{eqnarray}

Note that the random process $\nr(z)\ni(z)$ is mean-square integrable in $[0,\length]$ if and only if the integration
\begin{equation}\label{ch3.3.86}
{\int_{0}^{\length}\int_{0}^{\length}\mathrm{E}\left[\nr(s)\ni(s)\nr(t)\ni(t)\right]\d s\d t<\infty.}
\end{equation}
According to \eqref{ch3.3.60} and \eqref{ch3.3.104}--\eqref{ch3.3.108}, if $\mathrm{E}\beta(0)<\infty$, and $\mathrm{E}[\beta(0)]^{2}<\infty$, we have
\begin{eqnarray}\label{ch3.3.87}
\int_{0}^{\length}\int_{0}^{\length}\mathrm{E}\left[\nr(s)\ni(s)\nr(t)\ni(t)\right]\d s\d t&=& 2\iint_{\{s\leq t\}}\mathrm{E}\left[\nr(s)\ni(s)\nr(t)\ni(t)\right]\d\sigma\nonumber\\
&=&2\iint_{\{s\leq t\}}\frac{1}{12}\epsilon^{4}s^{2}\mathrm{E}\left[\beta(0)^{2}\right]+\frac{1}{18}\epsilon^{6}s^{3}\mathrm{E}\beta(0)+\nonumber\\
&&\frac{1}{144}\epsilon^{8}s^{4}+0+0+\frac{1}{4}\epsilon^{4}st\left(\frac{1}{6}\epsilon^{2}s\mathrm{E}\beta(0)+\frac{1}{24}\epsilon^{4}s^{2}\right)\d\sigma\nonumber\\
&=&\frac{1}{72}\epsilon^{4}\length^{4}\mathrm{E}\left[\beta(0)^{2}\right]+\frac{1}{90}\epsilon^{6}\length^{5}\mathrm{E}\beta(0)+\frac{23}{17280}\epsilon^{8}\length^{6}\label{ch3.3.109}\\
&<&\infty.\nonumber
\end{eqnarray}
As a result, the stochastic process $\nr(z)\ni(z)$ is mean-square integrable in $[0,\length]$.
}
\end{proof}

\section{Proof of Lemma \ref{ch3.l3.11}}\label{AppendixB3}

\newpage

\begin{proof}[Proof of Lemma \ref{ch3.l3.11}]
{1) We have
\begin{eqnarray}
&&\mathrm{E}\left[\underline{\Gamma^{(R)}(\length)}\underline{\Gamma^{(I)}(\length)}\Big|\underline{\zeta(0)}=\zeta(0)\right]\nonumber\\
&=&\mathrm{E}\left[\int_{0}^{\length}\int_{0}^{\length}\underline{\nr(s)}\underline{\ni(t)}\d s\d t\Bigg|\underline{\zeta(0)}=\zeta(0)\right]\nonumber\\
&=&\int_{0}^{\length}\int_{0}^{\length}\mathrm{E}\left[\underline{\nr(s)}\underline{\ni(t)}\Big|\underline{\zeta(0)}=\zeta(0)\right]\d s\d t,\label{ch3.a.99}
\end{eqnarray}
where we calculate $\mathrm{E}\left[\underline{\nr(s)}\underline{\ni(t)}\Big|\underline{\zeta(0)}=\zeta(0)\right]$. Without loss of generality, we assume that $s\leq t$. So we have 
\begin{eqnarray}
&&\mathrm{E}\left[\underline{\nr(s)}\underline{\ni(t)}\Big|\underline{\zeta(0)}=\zeta(0)\right]\nonumber\\
&=&\mathrm{E}\left\{\underline{\nr(s)}\left[\underline{\ni(s)}+\underline{\ni(t)}-\underline{\ni(s)}\right]\Big|\underline{\zeta(0)}=\zeta(0)\right\}\nonumber\\
&=&\mathrm{E}\left[\underline{\nr(s)}\underline{\ni(s)}\Big|\underline{\zeta(0)}=\zeta(0)\right]+\mathrm{E}\left\{\underline{\nr(s)}\left[\underline{\ni(t)}-\underline{\ni(s)}\right]\Big|\underline{\zeta(0)}=\zeta(0)\right\}\nonumber\\
&=&0+\mathrm{E}\left\{\underline{\nr(s)}\left[\underline{\nu_{I}(t)}-\underline{\nu_{I}(s)}+\frac{1}{2}\epsilon^{2}(t-s)\right]\Big|\underline{\zeta(0)}=\zeta(0)\right\},\label{ch3.a.100}
\end{eqnarray}
\eqref{ch3.a.100} because of Theorem \ref{ch3.t3.5}. Similar to the proof of Theorem \ref{ch3.t3.5}, we can also show that
\begin{equation}\label{ch3.a.101}
{\mathrm{E}\left\{\underline{\nr(s)}\left[\underline{\nu_{I}(t)}-\underline{\nu_{I}(s)}\right]\Big|\underline{\zeta(0)}=\zeta(0)\right\}=0.}
\end{equation}
Plug \eqref{ch3.a.100}--\eqref{ch3.a.101} into \eqref{ch3.a.99}, we have
$$\mathrm{E}\left[\underline{\Gamma^{(R)}(\length)}\underline{\Gamma^{(I)}(\length)}\Big|\underline{\zeta(0)}=\zeta(0)\right]=0.$$

2) We have 
\begin{eqnarray}
&&\mathrm{E}\left[\underline{\Gamma^{(R)}(\length)}\underline{\Gamma^{(RI)}(\length)}\Big|\underline{\zeta(0)}=\zeta(0)\right]\nonumber\\
&=&\mathrm{E}\left[\int_{0}^{\length}\int_{0}^{\length}\underline{\nr(s)}\underline{\nr(t)}\underline{\ni(t)}\d s\d t\Bigg|\underline{\zeta(0)}=\zeta(0)\right]\nonumber\\
&=&\int_{0}^{\length}\int_{0}^{\length}\mathrm{E}\left[\underline{\nr(s)}\underline{\nr(t)}\underline{\ni(t)}\Big|\underline{\zeta(0)}=\zeta(0)\right]\d s\d t\nonumber\\
&=&\int_{0}^{\length}\int_{0}^{\length}\mathrm{E}\left[\underline{\nr(s)}\underline{\nr(t)}\underline{\nu_{I}(t)}\Big|\underline{\zeta(0)}=\zeta(0)\right]\d s\d t+\nonumber\\
&&\int_{0}^{\length}\int_{0}^{\length}\frac{1}{2}\epsilon^{2}t\cdot\mathrm{E}\left[\underline{\nr(s)}\underline{\nr(t)}\Big|\underline{\zeta(0)}=\zeta(0)\right]\d s\d t.\label{ch3.a.102}
\end{eqnarray}
For any $(s,t)\in[0,\length]\times[0,\length]$, similar to the discussion in obtaining \eqref{ch3.3.69}, and using Lemma \ref{ch3.l3.9}, we have
\begin{equation}
\mathrm{E}\left[\underline{\nr(s)}\underline{\nr(t)}\underline{\nu_{I}(t)}\Big|\underline{\zeta(0)}=\zeta(0)\right]=\mathrm{E}\left[\underline{\nr(w)}^{2}\underline{\nu_{I}(w)}\Big|\underline{\zeta(0)}=\zeta(0)\right]=0,\label{ch3.a.103}
\end{equation}
and
\begin{equation}
\mathrm{E}\left[\underline{\nr(s)}\underline{\nr(t)}\Big|\underline{\zeta(0)}=\zeta(0)\right]=\mathrm{E}\left[\underline{\nr(w)}^{2}\Big|\underline{\zeta(0)}=\zeta(0)\right]=\frac{1}{6}\epsilon^{2}w\beta(0)+\frac{1}{24}\epsilon^{4}w^{2},\label{ch3.3.112}
\end{equation}
where $w\triangleq\min\{s,t\}$. Plug \eqref{ch3.a.103} and \eqref{ch3.3.112} back to \eqref{ch3.a.102}, we have 
\begin{eqnarray}
\mathrm{E}\left[\underline{\Gamma^{(R)}(\length)}\underline{\Gamma^{(RI)}(\length)}\Big|\underline{\zeta(0)}=\zeta(0)\right]&=&\frac{1}{2}\epsilon^{2}\bigg\{\iint_{\{s\leq t\}}\left[\frac{1}{6}\epsilon^{2}s\beta(0)+\frac{1}{24}\epsilon^{4}s^{2}\right]t\d\sigma_{1}+\nonumber\\
&&\iint_{\{s>t\}}\left[\frac{1}{6}\epsilon^{2}t\beta(0)+\frac{1}{24}\epsilon^{4}t^{2}\right]t\d\sigma_{2}\bigg\}\nonumber\\
&=&\frac{5}{288}\epsilon^{4}\length^{4}\beta(0)+\frac{7}{2880}\epsilon^{6}\length^{5}.\label{appen.b.4.4}
\end{eqnarray}

3) We calculate
\begin{eqnarray}
\mathrm{E}\left[\underline{\ni(s)}\underline{\nr(t)}\underline{\ni(t)}\Big|\underline{\zeta(0)}=\zeta(0)\right]&=&\mathrm{E}\left\{\underline{\nr(t)}\prod_{x\in\{s,t\}}\left[\underline{\nu_{I}(x)}+\frac{1}{2}\epsilon^{2}x\right]\Big|\underline{\zeta(0)}=\zeta(0)\right\}\nonumber\\
&=&\mathrm{E}\left[\underline{\nu_{I}(s)}\underline{\nr(t)}\underline{\nu_{I}(t)}\Big|\underline{\zeta(0)}=\zeta(0)\right].\label{ch3.3.113}
\end{eqnarray}
When $s\geq t$, we have 
\begin{eqnarray}
&&\mathrm{E}\left[\underline{\nu_{I}(s)}\underline{\nr(t)}\underline{\nu_{I}(t)}\Big|\underline{\zeta(0)}=\zeta(0)\right]\nonumber\\
&=&\mathrm{E}\left\{\underline{\nr(t)}\underline{\nu_{I}(t)}\left[\underline{\nu_{I}(t)}+\underline{\nu_{I}(s)}-\underline{\nu_{I}(t)}\right]\Big|\underline{\zeta(0)}=\zeta(0)\right\}\nonumber\\
&=&\mathrm{E}\left[\underline{\nr(t)}\underline{\nu_{I}(t)}^{2}\Big|\underline{\zeta(0)}=\zeta(0)\right]+\mathrm{E}\left\{\underline{\nr(t)}\underline{\nu_{I}(t)}\left[\underline{\nu_{I}(s)}-\underline{\nu_{I}(t)}\right]\Big|\underline{\zeta(0)}=\zeta(0)\right\}\nonumber\\
&=&\mathrm{E}\left[\underline{\nr(t)}\underline{\nu_{I}(t)}^{2}\Big|\underline{\zeta(0)}=\zeta(0)\right].\label{ch3.3.114}
\end{eqnarray}
When $s<t$, similar to the discussion in \eqref{ch3.3.114}, we have
\begin{equation}
\mathrm{E}\left[\underline{\nu_{I}(s)}\underline{\nr(t)}\underline{\nu_{I}(t)}\Big|\underline{\zeta(0)}=\zeta(0)\right]=\mathrm{E}\left[\underline{\nr(s)}\underline{\nu_{I}(s)}^{2}\Big|\underline{\zeta(0)}=\zeta(0)\right].\label{ch3.3.115}
\end{equation}
Plug \eqref{ch3.3.114}--\eqref{ch3.3.115} back to \eqref{ch3.3.113}, according to Lemma \ref{ch3.l3.9}, we have
\begin{equation}
\mathrm{E}\left[\underline{\ni(s)}\underline{\nr(t)}\underline{\ni(t)}\Big|\underline{\zeta(0)}=\zeta(0)\right]=\mathrm{E}\left[\underline{\nr(w)}\underline{\nu_{I}(w)}^{2}\Big|\underline{\zeta(0)}=\zeta(0)\right]=0,\label{ch3.3.116}
\end{equation}
where $w=\min\{s,t\}$. So we have
\begin{equation}
\mathrm{E}\left[\underline{\Gamma^{(I)}(\length)}\underline{\Gamma^{(RI)}(\length)}\Big|\underline{\zeta(0)}=\zeta(0)\right]=\int_{0}^{\length}\int_{0}^{\length}\mathrm{E}\left[\underline{\ni(s)}\underline{\nr(t)}\underline{\ni(t)}\Big|\underline{\zeta(0)}=\zeta(0)\right]\d s\d t=0.\label{ch3.3.117}
\end{equation}
}
\end{proof}

\section*{Acknowledgment}

This work was supported in part by ARC Discovery Project DP150103658.

\bibliographystyle{bibliography/IEEEtran}
\end{document}